\documentclass[11pt,leqno]{article}
\usepackage{amsmath,amssymb,amsfonts} 
\usepackage{epsfig} 
\usepackage{latexsym,nicefrac,bbm}
\usepackage{xspace}
\usepackage{color,fancybox,graphicx,fullpage}
\usepackage[top=1in, bottom=1in, left=1in, right=1in]{geometry}
\usepackage{tabularx} 
\usepackage{hyperref} 
\usepackage{pdfsync}
\usepackage[boxruled]{algorithm2e}
\usepackage{multicol}
\usepackage{enumitem}
\usepackage{subcaption}

\newcommand{\parag}[1]{\noindent {\bf #1}}

\newcommand{\nfrac}{\nicefrac}
\newcommand{\eps}{\epsilon}

\renewcommand{\epsilon}{\varepsilon}
\renewcommand{\tilde}{\widetilde}
\renewcommand{\hat}{\widehat}

\newtheorem{theorem}{Theorem}[section]

\newtheorem{fact}[theorem]{Fact}

\newenvironment{proof}{\begin{trivlist} \item {\bf Proof:~~}}
   {\qed\end{trivlist}}
\newenvironment{proofof}[1]{\begin{trivlist} \item {\bf Proof
#1:~~}}
  {\qed\end{trivlist}}

\def\FullBox{\hbox{\vrule width 6pt height 6pt depth 0pt}}

\def\qed{\ifmmode\qquad\FullBox\else{\unskip\nobreak\hfil
\penalty50\hskip1em\null\nobreak\hfil\FullBox
\parfillskip=0pt\finalhyphendemerits=0\endgraf}\fi}

\def\qedsketch{\ifmmode\Box\else{\unskip\nobreak\hfil
\penalty50\hskip1em\null\nobreak\hfil$\Box$
\parfillskip=0pt\finalhyphendemerits=0\endgraf}\fi}

\newcommand\Z{\mathbb Z}
\newcommand\N{\mathbb N}
\newcommand\R{\mathbb R}

\newcommand{\inparen}[1]{\left(#1\right)}             %\inparen{x+y}  is (x+y)
\newcommand{\inbraces}[1]{\left\{#1\right\}}           %\inbrace{x+y}  is {x+y}
\newcommand{\inangle}[1]{\left\langle#1\right\rangle} %\inangle{A}    is <A>

\newcommand{\poly}{\mathrm{poly}}

\newcommand{\cB}{\mathcal{B}}

\newcommand{\OPT}{\mathrm{OPT}}
\newcommand{\NP}{\mathbf{NP}}

\newcommand{\rOmega}{\widetilde{\Omega}}

\title{\bf Ranking with Fairness Constraints}

\date{}
\usepackage{authblk}

\author{L. Elisa Celis}
\author{Damian Straszak}
\author{Nisheeth K. Vishnoi}
\affil{\small \'{E}cole Polytechnique F\'{e}d\'{e}rale de Lausanne (EPFL), Switzerland}

\begin{document}
\maketitle

\begin{abstract}
Ranking algorithms are deployed widely to order a set of items in applications such as search engines, news feeds, and recommendation systems.
Recent studies, however, have shown that, left unchecked, the output of ranking algorithms can result in decreased diversity in the type of content presented,   promote stereotypes, and polarize opinions.
In order to address such issues, we study the following variant of the traditional ranking problem when, in addition, there are  {\em fairness} or {\em diversity} constraints.
Given a collection of items along with 1) the value of placing an item in a particular position in the ranking, 2) the collection of {\em sensitive} attributes (such as gender, race, political opinion) of each item and 3) a collection of fairness constraints that, for each $k$, bound the number of items with each attribute that are allowed to appear in the top $k$ positions of the ranking, the goal is to output a ranking that maximizes the value with respect to the original rank quality metric while respecting the constraints.
This problem encapsulates various well-studied problems related to bipartite and hypergraph matching as special cases and turns out to be hard to approximate even with simple constraints.
Our main technical contributions are fast exact and approximation algorithms along with complementary hardness results that, together, come close to settling the approximability of this constrained ranking maximization problem.
Unlike prior work on the approximability of constrained matching problems, our algorithm runs in linear time, even when the number of constraints is (polynomially) large, its approximation ratio does not depend on the number of constraints, and it produces solutions with small constraint violations.
Our results rely on insights about the constrained matching problem when the objective function satisfies certain properties that appear in common ranking metrics such as discounted cumulative gain (DCG), Spearman's rho or Bradley-Terry, along with the nested structure of fairness constraints.
\end{abstract}

\newpage

\tableofcontents

\newpage

\section{Introduction}
Selecting and ranking a subset of data is a fundamental problem in information retrieval and at the core of ubiquitous applications including ordering search results such (e.g., Google), personalized social media feeds (e.g., Facebook, Twitter or Instagram), ecommerce websites (e.g., Amazon or eBay), and online media sites (e.g., Netflix or YouTube). 
The basic algorithmic problem that arises is as follows: There are $m$ \emph{items} (e.g., webpages, images, or  documents), and the goal is to output a list of $n \ll m$ items in the order that is most {\em valuable} to a given user or company. 
For each item $i \in [m]$ and a position $j \in [n]$ one is given a number $W_{ij}$ that captures the \emph{value} that item $i$ contributes to the ranking if placed at position $j$. 
These values can be tailored to a particular query or user and a significant effort has gone into developing models and mechanisms to learn these parameters \cite{IRbook}. 
In practice there are many ways one could arrive at $W_{ij}$, each of which results in a slightly different metric for the value of a ranking -- prevalent examples include versions of discounted cumulative gain (DCG) \cite{jarvelin2002cumulated}, Bradley-Terry \cite{bradley1952rank} and Spearman's rho \cite{spearman1904proof}. 
Note that for many of these metrics, one does not necessarily need $nm$ parameters to specify $W$ and typically $m$ ``degrees of freedom'' is enough (just specifying the ``quality of each item'').
Still, we choose to work with this general setting, and only abstract out the most important properties such a weight matrix $W$ satisfies.
Generally, for such metrics, $W_{ij}$ is non-increasing in both $i$ and $j$, and 
if we interpret $i_1 < i_2$ to mean that $i_1$ has better {\em quality} than $i_2$, then the value of the ranking can only increase by placing $i_1$ above $i_2$ in the ranking. 
Formally, such values satisfy the following  property (known as monotonicity and the Monge condition)
\begin{equation} 
\label{eq:monge} 
{W_{i_1j_1} \geq W_{i_2j_1} \; \; \; \mathrm{and} \; \; \;  W_{i_1j_1} \geq W_{i_1j_2}  \; \; \; \mathrm{and} \; \; \;  W_{i_1 j_1}+W_{i_2 j_2} \geq W_{i_1j_2}+W_{j_1i_2}}
\end{equation}
for all $1\leq i_1 < i_2 \leq m$ and $1\leq j_1 < j_2 \leq n$ (see Appendix~\ref{sec:metrics}). 
The {\em ranking maximization} problem is to find an assignment of the items to each of the $n$ positions such that  the total value obtained is maximized.
In this form, the problem is equivalent to finding the maximum weight matching in a complete $m \times n$ bipartite graph and has a well known solution -- the Hungarian algorithm.
However, recent studies have shown that producing rankings in this manner 
can result in one type of content being overrepresented at the expense of another. 
This is a form of {\em algorithmic bias}  and can lead to grave societal consequences --  
from search results that inadvertently promote stereotypes by over/under-representing sensitive attributes such as race and gender \cite{kay2015unequal,bolukbasi2016man}, 
to news feeds that can promote extremist ideology \cite{costello2016views} and possibly even influence the results of elections \cite{Baer2016,bakshy2015exposure}.
For example, \cite{Epstein2015} demonstrated that by varying the ranking of a set of news articles the voting preferences of undecided voters can be manipulated.  
Towards ensuring that no type of content is overrepresented  
in the context of the ranking problem as defined above, we introduce the {\em constrained ranking maximization} problem that restricts allowable rankings to those in which no type of content dominates -- i.e., {\em to ensure the rankings are fair}.

Since fairness (and bias) could mean different things in different contexts, rather than fixing one specific notion of fairness, 
we allow the \emph{user} to specify a set of \emph{fairness constraints}; in other words, we take the constraints as input.
As a motivating example, consider the setting in which the set of items consists of $m$ images of computer scientists, each image is associated with several (possibly non-disjoint) sensitive attributes or \emph{properties} such as gender, ethnicity and age, and a subset of size $n$ needs to be selected and ranked. 
The user can specify an upper-bound $U_{k\ell} \in \mathbb Z_{\geq 0}$  on the number of items with property $\ell$ that are allowed to appear in the top $k$ positions of the ranking, and similarly a lower-bound $L_{k\ell}$. 
Formally, let $\{1, 2, \ldots, p\}$ be a set of properties  
and let $P_\ell \subseteq [m]$ be the set of items that have the property $\ell$ (note that these sets need not be disjoint). 
Let $x$ be an $m \times n$ binary assignment matrix  whose $j$-th column contains a one in the $i$-th position if item $i$ is assigned to position $j$ (each position must be assigned to exactly one item and each item can be assigned to at most one position).
We say that $x$ satisfies the fairness constraints if for all $\ell \in [p]$ and $k \in [n]$,  we have
\begin{equation*}
{ L_{k\ell}\leq \sum\limits_{1 \leq j \leq k}  \sum\limits_{i \in P_{\ell}} x_{ij} \leq U_{k\ell},}
\end{equation*}
If we let $\cB$ be the family of all assignment matrices $x$ that satisfy the fairness constraints, the constrained ranking optimization problem is: Given the sets of items with each property $\{P_1,\ldots,P_p\}$, the fairness constraints, $\{L_{k\ell}\}$, $\{U_{k\ell}\}$, and the values $\{W_{ij}\}$,  find 
\begin{equation*}
{\mathop{\arg\max}\limits_{x \in \cB} \sum\limits_{i \in [m], j \in [n]} W_{ij} x_{ij}.} 
\end{equation*}
This problem is equivalent to finding a maximum weight matching of size $n$ \emph{that satisfies the given fairness constraints} in a weighted complete $m\times n$ bipartite graph, and now becomes non-trivial -- its complexity is the central object of study in this paper.

Beyond the fairness and ethical considerations, traditional {\em diversification} concerns in information retrieval such as {query ambiguity} (does ``jaguar''  refer to the car or the animal?) or {user context} (does the user want to see webpages, news articles, academic papers or images?) can also be cast in this framework.
Towards this, a rich literature on diversifying rankings has emerged in information retrieval.  
On a high-level, several approaches  redefine the objective function to incorporate a notion of diversity and leave the ranking maximization problem unconstrained. E.g., a common approach is to re-weight the $w_{ij}$s to attempt to capture the amount of diversity item $i$ would introduce at position $k$ conditioned on the items that were placed at positions $1, \ldots, k-1$ (see \cite{carbonell1998use,ZH2008,clarke2008, ZCL2003, ZMKL2005}), or 
casting it directly as an (unconstrained) multi-objective optimization problem \cite{yang2016measuring}. 
Alternate approaches mix together or aggregate different rankings, e.g., as generated by different interpretations of a query \cite{radlinski2009redundancy,dwork2001rank}.
Diversity has also been found to be desirable by users \cite{CelisTekriwal2017}, and has been observed to arise inherently when the ranking is determined by user upvotes \cite{CelisKK16}
Despite these efforts and the fact that all major search engines  now diversify their results, highly uniform content is often still displayed -- e.g., certain image searches can display results that have almost entirely the same attributes  \cite{kay2015unequal}.
Further,  \cite{gollapudi2009axiomatic}  showed that no  {single} diversification function can satisfy a set of natural axioms that one would want any fair ranking to have. 
In essence, there is a tension between relevance and fairness -- if the $w_{ij}$s for items that have a given property are much higher than the rest, the above approaches cannot correct for overrepresentation. 
Hence the reason to cast the problem as a constrained optimization problem: The objective is still determined by the values but the solution space is restricted by fairness constraints.

Theoretically, the fairness constraints come with a computational price:  
The constrained ranking maximization problem can be seen to generalize various $\NP$-hard problems such as independent set, hypergraph matching and set packing.
Unlike the unconstrained case, even checking  if there is a complete feasible ranking (i.e., $\cB \neq \emptyset$) is $\NP$-hard. 
As a consequence, in general, we cannot hope to produce a solution that does not violate any constraints.  
Some variants and generalizations of our problem have been studied in the TCS and optimization literature; here we mention the three most relevant. Note that some may leave empty positions in the ranking as opposed to selecting $n$ elements to rank as we desire.
\cite{AFK96} considered the  bipartite perfect matching problem with $\poly(m)$ constraints.
They present a polynomial time randomized algorithm that finds a near-perfect matching which violates each constraint additively by at most $O(\sqrt m)$.
\cite{GRSZ13} improved the above result to a $(1+\eps)$-approximation algorithm; however, the running time of their algorithm is roughly $ m^{\nfrac{K^{2.5}}{{\eps^2}}}$ where $K$ is the number of hard constraints and the output is a  matching. 
\cite{Srinivasan95} studied the approximability of  the packing integer program problem which, when applied to our setting and gives an $O(\sqrt{m})$ approximation algorithm.  
In the constrained ranking maximization problem presented above, all of these results seem inadequate; the number of fairness constraints is $2np$ which would make the running time of \cite{GRSZ13}  too large and an additive violation of $O(\sqrt{m})$ would render the upper-bound constraints impotent. 
The main technical contributions of this paper are fast, exact and approximation algorithms for this constrained ranking maximization problem along with complementary hardness results which, together, give a solid understanding of the computational complexity of this problem.
To overcome the limitations of the past work on constrained matching problems, our results often make use of two structural properties of such a formulation:
A) The set of constraints can be broken into $p$ groups; for each property $\ell \in [p]$ we have  $n$ (nested) upper-bound constraints, one for each $k \in [n]$, and 
B) The objective function satisfies the  property stated in \eqref{eq:monge}.
Using properties A) and B) we obtain efficient -- polynomial, or even linear time algorithms for this problem in various interesting regimes.
Both these properties are natural in the information retrieval setting and could be useful in other algorithmic contexts involving rankings.
%

%\newpage

\section{Our Model}\label{sec:model}
We study the following  {\em constrained ranking maximization problem}
\begin{equation}
{\mathop{\arg\max}\limits_{x\in R_{m,n}} \sum\limits_{i \in [m], j \in [n]} W_{ij} x_{ij}} ~~~~~~
\mathrm{s. t.} ~~~~
 L_{k\ell} \leq \sum\limits_{1 \leq j \leq k}  \sum\limits_{i \in P_{\ell}} x_{ij} \leq U_{k\ell}~~~~~\forall~ \ell \in [p],k\in [n],
\end{equation}
where $R_{m,n}$ is the set of all matrices $\{0,1\}^{m \times n}$ which represent ranking $m$ items into $n$ positions. Recall that $W_{ij}$ represents the profit of placing item $i$ at position $j$ and for every property $\ell \in [p]$ and every position $k$ in the ranking, $L_{k\ell}$ and $U_{k\ell}$ are the lower and upper bound on the number of items having property $\ell$ that are allowed to be in the top $k$ positions in the ranking.
For an example, we refer to Figure~\ref{fig:example}.

We distinguish two important special cases of the problem: when only the upper-bound constraints are present, and when only the lower-bound constraints are present.
These variants are referred to as the {\em constrained ranking maximization problem (U)} and the {\em constrained ranking maximization problem (L)} respectively, and to avoid confusion we sometimes add (LU) when talking about the general problem with both types of constraints present.
Furthermore, most our results hold under the assumption that the weight function $W$ is monotone and satisfies the Monge property~\eqref{eq:monge}, whenever these assumptions are not necessary, we emphasize this fact by saying that {\it general weights} are allowed.

\begin{figure}[t!]
\centering
%\hspace{-.25in}
\subcaptionbox{\footnotesize  An example of a value matrix $W$. The values corresponding to the optimal (unconstrained) ranking in (b) and the optimal constrained ranking in (c) are depicted by gray and orange  respectively. 
(Note that, for clarity of the rank order, the above is the transpose of the matrix referred to as $W$ in the text.)
\label{sub-W}}{%
  \includegraphics[width=.48\columnwidth]{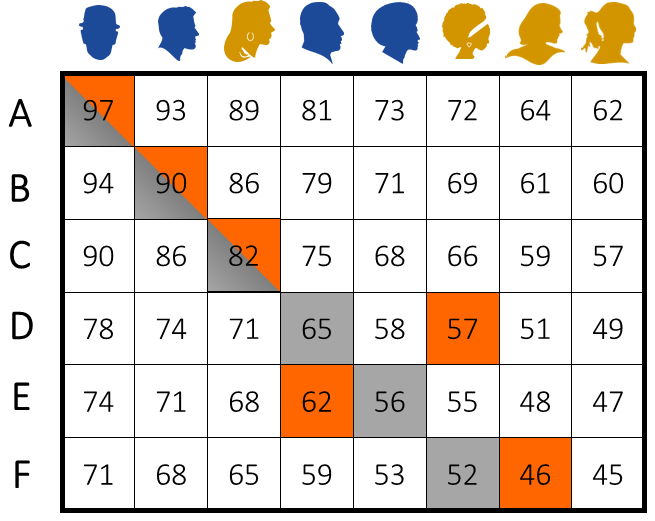}%
}\hfil
\subcaptionbox{\footnotesize The optimal unconstrained ranking. The upper-bound constraint at position D is violated as there are 3 men, but only 2 are allowed.\label{sub-unfair}}{%
  \includegraphics[width=.21\columnwidth]{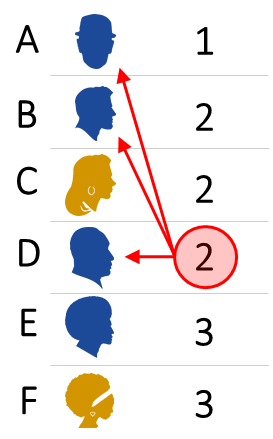}%
}\hfil
\subcaptionbox{\footnotesize The optimal constrained ranking. The upper-bound constraint at position D is no longer violated; in fact all constraints are satisfied.\label{sub-fair}}{%
  \includegraphics[width=.21\columnwidth]{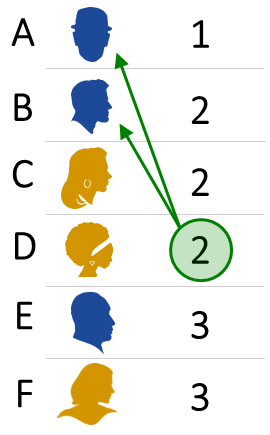}%
}
%\hspace{-.25in}
%\vspace{-.05in}
\caption{A simple example of our framework: In (a) a matrix of $W_{ij}$s is presented. Here, the options are people who are either male (blue) or female (yellow), and 6 of them must be ranked. We assume that there is a single upper-bound constraint for each position in the ranking which is applied to both genders as depicted in figures (b) and (c). The constraints are satisfied in the latter, but not the former. The weights of these two rankings are depicted in figure (a).}\label{fig:example}
%\vspace{-.05in}
\end{figure}
\section{Our  Results}
In this section, we present an overview of our results. 
The statements of theorems here are informal for the ease of readability, formal statements are deferred to a later part.

Let the \emph{type} 
$$T_i:=\{ \ell \in [p] : i \in P_\ell \}$$ of item $i$ be the set of properties that the item $i$ has. 
Our first result is an exact algorithm for solving the constrained ranking  maximization problem whose running time is polynomial if the number of distinct $T_i$s, denoted by $q$, is constant; see Theorem~\ref{thm:dp_formal}.
\begin{theorem}[Exact dynamic programming-based algorithm]\label{thm:dp}
There is an algorithm that solves the constrained ranking maximization problem (LU) in $O(pqn^q+pm)$ time 
 when the values $W$ satisfy property \eqref{eq:monge}. 
\end{theorem}
\noindent
This algorithm combines a geometric interpretation of our problem along with dynamic programming and proceeds by  solving a sequence of $q-$dimensional sub-problems. 
The proof of Theorem~\ref{thm:dp} is provided in Section~\ref{sec:dp}.
When $q$ is allowed to be large, the problem is $\NP$-hard; see Theorem~\ref{thm:hardnessAll}.

Generally, we may not be able to assume that $q$ is a constant and, even then, it would be desirable to have algorithms whose running time is close to $(m+n)p$, the size of the input.
Towards this we consider a natural parameter of the set of properties: The size of the largest $T_i$, namely 
$$ \Delta := \max_{i \in [m]} |T_i|.$$ 
The complexity of the constrained ranking maximization problem turns out to  show interesting behavior with respect to $\Delta$  (note that $\Delta \leq p$ and typically $p \ll q$).
The case when $\Delta =1$ corresponds to the simplest practical setting where there are $p$ disjoint properties, i.e., the properties partition the set of items.
For instance, a set of images of humans could be partitioned based on the ethnicity or age of the individual.
Note that even though $q=p$ for $\Delta = 1$, this $q$ could still be large and the previous theorem may have a prohibitively large running time.  

When $\Delta=1$ we prove that the constrained ranking maximization problem (LU) is polynomial time solvable.
\begin{theorem}[Polynomial time algorithm for $\Delta=1$]\label{thm:flow}
The constrained ranking maximization problem (LU)  for $\Delta=1$ can be solved in $\widetilde{O}(n^2m)$ time when the values $W$ satisfy property \eqref{eq:monge}.
\end{theorem} 
The above is obtained by reducing this variant of the ranking maximization problem  to the minimum cost flow problem, that can be solved efficiently (the network is acyclic).
We note that even though the running time is polynomial in $m$, it might be still not satisfactory for practical purposes.
With the aim of designing faster -- linear time algorithms, we focus on the case when only upper-bound constraints are present.
For this case, we analyze a natural linear programming (LP) relaxation for the constrained ranking maximization problem (U).
It reveals interesting structure of the problem and motivates a fast greedy algorithm.
Formally,  
the relaxation considers the  set $\Omega_{m,n}$ defined as 
$$\Omega_{m,n}:=\inbraces{ x \in {[0,1]}^{m \times n}:\sum_{j=1}^n x_{ij} \leq 1 ~\mbox{for all }i\in [m], ~~~
\sum_{i=1}^m x_{ij}=1,\mbox{ for all }j\in [n]}$$
and the following linear program
\begin{equation}\label{eq:lp}	
	\mathop{\max}_{x\in \Omega_{m,n}}  ~~  \sum_{i=1}^m \sum_{j=1}^n W_{ij} x_{ij} ~~~~
		 ~~~~ \mathrm{s.t.} ~~~~   \sum_{i\in P_\ell}\sum_{j=1}^k x_{ij}\leq U_{k\ell},  \mbox{ $\forall$ $\ell \in [p]$, $k\in [n]$.}
\end{equation}
Observe that in the absence of fairness constraints,~\eqref{eq:lp} represents the maximum weight bipartite matching problem -- it is well known that the feasible region of its fractional relaxation has integral vertices and hence the optimal values of these two coincide. 
However, in the constrained setting, even for $\Delta=1$, it can be shown that the feasible region is no longer integral -- it can have fractional vertices  (see Fact~\ref{non_integral}). 
For this reason, it is not true that maximizing any linear objective results in an integral solution.  
Surprisingly, we prove that for $\Delta=1$ the cost functions we consider are special and never yield optimal fractional (vertex) solutions. 
\begin{theorem}[Exact LP-based  algorithm for $\Delta=1$\label{thm:single}; see Theorems \ref{thm:lp1_formal}  and \ref{thm:FastGreedy}]
Consider the linear programming relaxation~\eqref{eq:lp} for the constrained ranking  maximization problem (U) when $\Delta=1$ and the objective function satisfies \eqref{eq:monge}. Then there exists an optimal solution with integral entries and hence the relaxation is exact. Further, there exists a greedy algorithm to find an optimal integral solution in  $O(np + m)$ time.
\end{theorem}
\noindent
The proof relies on a combinatorial argument on the structure of tight constraints that crucially uses  the assumption that $\Delta=1$ and the property \eqref{eq:monge} of the objective function. 
 Note that the result of Theorem~\ref{thm:single} implies in particular that whenever the linear program~\eqref{eq:lp} is feasible then there is also an integer solution -- a feasible ranking.
This can be also argued for the general (LU) variant of the problem and its corresponding LP relaxation.
However, extending Theorem~\ref{thm:single} to this case seems  more challenging and is left as an open problem.

When trying to design algorithms for larger $\Delta$, the   difficulty is that the constrained ranking {\em feasibility} problem remains $\NP$-hard (in fact, even hard to approximate when feasibility is guaranteed) for $\Delta\geq 3$; see Theorems \ref{thm:np-hard} and \ref{thm:hard_approximate}.
Together, these results imply that unless we restrict to feasible instances of the constrained ranking problem, it is impossible to obtain any reasonable approximation algorithm for this problem.
In order to bypass this  barrier, we focus on the (U) variant of the problem and present an {\em algorithmically verifiable} condition for feasibility and argue that it is natural in the context of information retrieval.
For each $1 \leq k \leq n,$ we consider the set 
$$ S_k:= \{ l \in [p] : U_{(k-1)\ell } +1 \leq  U_{k\ell} \}$$
 of all properties whose constraints increase by at least $1$ when going from the $(k-1)$st to the $k$th position. 
We observe that the following {\em abundance of items} condition is sufficient for feasibility:
 \begin{equation}\label{eq:feas1}
 \forall k \; \mbox{there are at least $n$ items} \; i \; \mathrm{s.t. } \; T_i  \subseteq S_k.
 \end{equation}
 Intuitively, this says that there should be always at least a few ways to extend a feasible ranking of $(k-1)$ items to a ranking of $k$ items.
 Simple examples show that this condition can be necessary for certain constraints $\{U_{k\ell}\}$.
In practice, this assumption is almost never a problem -- the available items $m$ (e.g., webpages) far outnumber the size of the ranking $n$ (e.g., number of results displayed in the first page) and the number of properties $p$ (i.e., there are only so many ``types'' of webpages).

We  show that assuming condition~\eqref{eq:feas1}, there is a linear-time algorithm that 
achieves an $(\Delta+2)$-approximation, while only slightly violating the upper-bound constraints.
This result does not need assumption \eqref{eq:monge}, rather only that the $W_{ij}$s are non-negative.
This result is near-optimal; we provide an $\Omega\inparen{\frac{\Delta}{\log \Delta}}$ hardness of approximation result (see Theorem~\ref{thm:hard_approximate}). 

\begin{theorem}[$(\Delta+2)$-approximation algorithm; see Theorem \ref{thm:mul_formal}]\label{thm:mult}
For the constrained ranking maximization problem (U), under the assumption \eqref{eq:feas1}, there is an algorithm that in linear time outputs a ranking $x$ with value at least $\frac{1}{\Delta+2}$ times the optimal one, such that $x$ satisfies the upper-bound constraints with at most a twice multiplicative violation, i.e., 
\[{\sum_{i \in P_\ell}\sum_{j=1}^k x_{ij}\leq 2U_{k\ell}, ~~\mbox{for all $\ell \in [p]$ and $k\in [n]$}.}\]
\end{theorem}
\noindent
One can construct artificial instances of the ranking problem, where the output of the algorithm indeed violates upper-bound constraints with a $2$-multiplicative factor.
However, these violations are caused by the presence of high-utility items with a large number of properties.
Such items are unlikely to appear in real-life instances and thus we expect the practical performance of the algorithm to be better than the worst-case bound given in Theorem~\ref{thm:mult} suggests. 
Lastly we summarize our hardness results for the constrained ranking problem.

\begin{theorem}[Hardness Results -- Informal]\label{thm:hardnessAll}
The following variants of the constrained ranking feasibility (U) and constrained ranking maximization (U) problem are $\NP$-hard.
\begin{enumerate}
%[topsep=1pt,itemsep=0ex,partopsep=0ex,parsep=1pt]
\item Deciding feasibility for the case of $\Delta\geq 3$  (Theorem~\ref{thm:np-hard}).
\item Under the feasibility condition~\eqref{eq:feas1}, approximating the optimal value of a ranking within a factor $O\inparen{\nfrac{\Delta}{\log \Delta}}$,  for any $\Delta \geq 3$ (Theorem~\ref{thm:hard_approximate}).
\item Deciding feasibility when only the number of items $m$, number of  positions $n$, and upper-bounds $u$ are given as input; the properties are fixed for every $m$ (Theorem~\ref{thm:fixedhardness}).
\item For every constant $c$, deciding between whether there exists a feasible solution or every solution violates some constraint by a factor of $c$ (Theorem~\ref{thm:violation}).
\end{enumerate} 
\end{theorem}

\section{Other Related Work}\label{sec:other_work}
Information retrieval, which focuses on selecting and ranking subsets of data, has a rich history in computer science, and is a well-established subfield in and of itself; see, e.g., the foundational work by \cite{Salton1988}. 
The probability ranking principle (PRP) forms the foundation of information retrieval research \cite{MK1960,Robe1977}; in our context it states that \emph{a system's ranking should order items by decreasing value.} 
Our problem formulation and solutions are in line with this -- \emph{subject to satisfying the diversity constraints}. 

A related problem is diverse data summarization in which a subset of items with varied properties must be selected from a large set \cite{panigrahi2012online,FatML}, or similarly, voting with diversity constraints in which a subset of items of people with varied properties or attributes must be selected via a voting procedure \cite{monroe1995fully,celis2017group}. However, the formulation of these problem is considerably different as there is no need to produce a ranking of the selected items, and hence the analogous notion of constraints is more relaxed. 
Extending work on fairness in classification problems \cite{zemel2013learning}, the fair ranking problem has also been studied as an (unconstrained) multi-objective optimization problem, and various fairness metrics  of a ranking have been proposed \cite{yang2016measuring}.

Combining the learning of values along with the ranking of items has also been studied \cite{radlinski2008learning,slivkins2013ranked}; in each round an algorithm chooses an ordered list of $k$ documents as a function of the estimated values $W_{ij}$ and can receive a click on one of them. These clicks are used to update the estimate of the $W_{ij}$s, and 
bounds on the regret (i.e., learning rate) can be given using a bandit framework. In this problem, while there are different types of items that can affect the click probabilities, there are no constraints on how they should be displayed.

Recent work has shown that, in many settings, there are impossibility results that prevent us from attaining both \emph{property} and \emph{item} fairness \cite{kleinberg_fatml}. Indeed, our work focuses on ensuring property fairness (i.e., no property is overrepresented), however this comes at an expense of item fairness (i.e., depending on which properties an item has, it may have much higher / lower probability of being displayed than another item with the same value). In our motivating application we deal with the ranking of documents or webpages, and hence are satisfied with this trade-off. However, further consideration may be required if, e.g., we wish to rank people as this would give individuals different likelihoods of being near the top of the list based on their properties rather than solely on their value.

\subsection*{Organization of the rest of the paper}
In Section~\ref{sec:other_work} we discuss other related work.
Section~\ref{sec:techoverview} contains an overview of the proofs of our main results.
Section~\ref{sec:flow} contains the proof of Theorem~\ref{sec:flow} -- a polynomial time algorithm for the $\Delta=1$ case.
The proof of Theorem~\ref{thm:dp} on the exact algorithm for general $\Delta$ is presented in Section~\ref{sec:dp}.
Section~\ref{sec:lp} contains the proof of Theorem~\ref{thm:single} which shows that there exists an integral solution to~\eqref{eq:lp}; it also provides a simple and fast greedy algorithm for the ranking maximization problem (U) for $\Delta=1$.
 Section \ref{sec:mul} contains the proof of Theorem \ref{thm:mult}, which gives our approximate result on the ranking maximization problem (U) for general $\Delta$.
Our hardness results are presented in Section~\ref{sec:hardness}. 
 In Section~\ref{sec:discussion} we provide a discussion of possible directions for future work and open problems.
 Finally in Appendix~\ref{sec:metrics} we give a brief overview of some common ranking metrics and explain how  they can be  captured by values $W_{ij}$ that satisfy \eqref{eq:monge}.

\section{Proof Overviews}
\label{sec:techoverview}

\medskip
\parag{Overview of the proof of  Theorem~\ref{thm:dp}.}
We first observe that the constrained ranking maximization problem has a  simple geometric interpretation. 
Every item $i\in [m]$ can be assigned a \emph{property vector} $t_i \in \{0,1\}^p$ whose $\ell$-th entry is $1$ if  item $i$ has property $\ell$ and $0$ otherwise.
We can then think of  the constrained ranking maximization problem as finding a sequence of $n$ distinct items $i_1, i_2, \ldots, i_n$ such that 
$L_k \leq \sum_{j=1}^k t_{i_j} \leq U_k$ for all $k\in [n]$,
where $U_k$ is the vector whose $\ell$-th entry is $U_{\ell k}$.
In other words, we require that the partial sums of the vectors corresponding to the top $k$ items in the ranking stay within the region $[L_{k1},U_{k1}] \times [L_{k2},U_{k2}] \times \cdots \times [L_{kn},U_{kn}]$ defined by the fairness constraints.

Let $Q:=\{t_i: i\in [m]\}$ be the set of all the different property vectors $t_i$ that appear for items $i\in [m]$, and let us denote its elements by $v_1, v_2, \ldots, v_q$.  
A simple but  important observation is that whenever two items $i_1, i_2\in[m]$ (with say $i_1<i_2$) have the same property vector: $t_{i_1} = t_{i_2}$, then in every optimal solution either $i_1$ will be ranked above $i_2$, only $i_1$ is ranked, or neither is used.
 This follows from the assumption that the weight matrix is monotone in $i$ and $j$ and satisfies the  property as stated in \eqref{eq:monge}. 

Let us now define the following sub-problem that asks for the \emph{property vectors} of a feasible solution: Given a tuple $(s_1, s_2, \ldots, s_q)\in \N^q$ such that $k=s_1+s_2+ \cdots + s_q \leq n$, what is the optimal way to obtain a feasible ranking on $k$ items such that $s_j$ of them have property vector equal to $v_j$ for all $j=1,2, \ldots, q$? 
Given a solution to this sub-problem, using the observation above, it is easy to determine which items should be used for a given property vector, and in what order.
Further, one can easily solve such a sub-problem given the solutions to smaller sub-problems (with a smaller sum of $s_j$s), resulting in a dynamic programming algorithm with $O(n^q)$ states and, hence, roughly the same running time. 

\medskip
\parag{Overview of the proof of  Theorem~\ref{thm:flow}.}
The main idea is to reduce ranking maximization to the minimum cost flow problem and then observe several structural properties of the resulting instance which allow one to solve it efficiently (in $\widetilde{O}(n^2m)$ time).

Given an instance of the constrained ranking maximization problem (U), we construct a weighted flow network $G=(V,E)$ such that every feasible ranking corresponds to a feasible flow of value $n$ in $G$.
Roughly, for every property $\ell$ a chain of $n$ vertices is constructed so that placing item $i$ (such that $i\in P_\ell$) at position $k$ corresponds to sending one unit of flow through the chain corresponding to item $i$ up to its $k$th vertex and then sending the flow to sink.
Edge weights in these gadgets (chains) are chosen in such a way that the cost of sending a unit through this path is $-W_{i,k}$.
The capacities in chains corresponding to properties implement upper-bound constraints.
The lower-bound constraints can be also enforced by putting appropriate weights on  edges of these chains.

The instance of the minimum cost flow problem we construct has $O(np)$ vertices and $O(nm)$ edges and is acyclic, which allows to replace the application of the Bellman-Ford algorithm in the first phase of the Successive Shortest Path algorithm by a linear-time procedure.
This then easily leads to an implementation in $O(n^2m \log m)$ time.

\medskip

\parag{Overview of the proof of  Theorem~\ref{thm:single}.} 
Unlike the $\Delta=0$ case where the LP-relaxation \eqref{eq:lp} has no non-integral vertex (it is the assignment polytope) , even when $p=1$, fractional vertices can arise (see Fact~\ref{non_integral}).
Theorem \ref{thm:single} implies that for $\Delta = 1$, although the feasible region of~\eqref{eq:lp} is not integral in all directions, it is along the directions of interest.
In the proof we first reduce the problem to the case when $m=n$ (i.e., when one has to rank all of the items) and $w$ has the {\em strict} form of  property \eqref{eq:monge} (i.e., when the inequalities in assumption \eqref{eq:monge} are strict). 
Our strategy  then is to prove that for every fractional feasible solution $x\in \Omega_{m,m}$ there is a direction $y\in \R^{m \times m}$ such that the solution $x^\prime:=x+\eps y$ is still feasible (for some $\eps>0$) and its weight is larger than the weight of $x$. 
This implies that every optimal solution is necessarily integral. 

Combinatorially, the directions we consider correspond to $4$-cycles in the underlying complete bipartite graph, such that the weight of the matching can be improved by swapping edges along the cycle. 
The argument that shows the existence of such a cycle makes use of the special structure of the constraints  in this family of instances. 

To illustrate the approach, suppose that there exist two items $i_1 < i_2$ that have the same property $\ell \in [p]$, and for some ranking positions $j_1<j_2$ we have 

\begin{equation}\label{eq:conf}
x_{i_1j_2}>0 ~~~~\mbox{and}~~~~x_{i_2 j_1}>0.
\end{equation}
\noindent
Following the strategy outlined above, consider $x^\prime=x+\eps y$ with $y\in \R^{m\times m}$ to be zero everywhere except $y_{i_1j_1}=y_{i_2j_2}=1$ and $y_{i_1 j_2}=y_{i_2j_1}=-1$.
 We would like to prove that the weight of $x^\prime$ is larger than the weight of $x$ and that $x^\prime$ is feasible for some (possibly small) $\eps>0$. 
The reason why we gain by moving in the direction of $y$ follows from  property \eqref{eq:monge}.
 Feasibility in turn follows because $y$ is orthogonal to every constraint defining the feasible region. 
  Indeed, the only constraints involving items $i_1, i_2$ are those corresponding to the property $\ell$. 
  Further, every such constraint is of the form\footnote{By $\inangle{\cdot, \cdot}$ we denote the inner product between two matrices, i.e., if $x, y \in \R^{m \times n}$ then $\inangle{x, y} := \sum_{j=1}^m \sum_{i=1}^n x_{ij} y_{ij}.$} $\inangle{1_{R_k}, x} \leq U_{k\ell}$ where $1_{R_k}$ is the indicator vector of a rectangle $R_k:=P_\ell\times  [k]$.
   Such a rectangle contains either all non-zero entries of $y$, two non-zero entries (with opposite signs), or none. In any of these cases, $\inangle{1_{R_k}, y} =0$. 

Using a reasoning as above, one can show that no configuration of the form~\eqref{eq:conf} can appear in any optimal solution for $i_1, i_2$ that share a property $\ell$. 
This implies that the support of every optimal solution has a certain structure when restricted to items that have any given property $\ell \in [p]$; this structure allows us to find an improvement direction in case the solution is not integral.
To prove integrality we show that for every fractional solution $x\in \R^{m\times m}$ there exists a fractional entry $x_{ij}\in (0,1)$ that can be slightly increased without violating the fairness constraints. 
Moreover since the $i$-th row and the $j$-th column must contain at least one more fractional entry each (since the row- and column-sums are $1$), we can construct  (as above) a direction $y$, along which the weight can be increased. 
The choice of the corresponding entries that should be altered requires some care, as otherwise we might end up violating fairness constraints. 

The second part of Theorem~\ref{thm:single} is an algorithm for solving the constrained ranking maximization problem for $\Delta=1$ in optimal (in the input size) running time of $O(np+m)$. 
We show that a natural greedy algorithm can be used. 
More precisely, one iteratively fills in ranking positions by always selecting the {\em highest value} item that is still available and does not lead to a constraint violation. 
 An inductive argument based that relies on property~\ref{eq:monge} and the $\Delta=1$ assumption gives the correctness of such a procedure.

\medskip

\parag{Overview of the proof of  Theorem~\ref{thm:mult}.} 
Let $\Delta>1$ be arbitrary. 
 The most important part of our algorithm is a greedy procedure that finds a large weight solution to a slightly relaxed problem in which not all positions in the ranking have to be occupied. 
 It processes pairs $(i,j) \in [m] \times [n]$ in non-increasing order of weights $W_{ij}$ and puts item $i$ in position $j$ whenever this does not lead to constraint violation. 
%

%%%
To analyze the approximation guarantee of this algorithm let us first inspect the combinatorial structure of the feasible set.
In total there are $p \cdot n$ fairness constraints in the problem and additionally $m+n$ ``matching'' constraints, saying that no ``column'' or ``row'' can have more than a single one in the solution matrix $x\in \{0,1\}^{m \times n}$. 
However, after relaxing the problem to the one where not all ranking positions have to be filled, one can observe that the feasible set is just an intersection of $p+2$ matroids on the common ground set $[m] \times [n]$.
Indeed, two of them correspond to the matching constraints, and are partition matroids.
The remaining $p$ matroids correspond to properties: for every property $\ell$ there is a chain of subsets $S_1 \subseteq S_2 \subseteq \cdots \subseteq S_n$ of $[m] \times [n]$ such that 
$$\mathcal{I}_{\ell}=\{S \subseteq [m] \times [n]: |S\cap S_k| \leq  U_{k\ell} \mbox{ for all }k=1,2, \ldots, n\}$$ is the set of independent sets in this (laminar) matroid.
In the work~\cite{Jenkyns76} it is shown that the greedy algorithm run on an intersection of $K$ matroids yields $K$-approximation, hence $(p+2)$-approximation of our algorithm follows.

To obtain a better -- $(\Delta+2)$-approximation bound, a more careful analysis is required. 
The proof is based on the fact that, roughly, if a new element is added to a feasible solution $S$, then at most $\Delta+2$ elements need to be removed from $S$ to  make it again feasible. 
Thus adding greedily one element can cost us absence of $\Delta+2$ other elements of weight at most the one we have added.
This idea can be formalized and used to prove the $(\Delta+2)$-approximation of the greedy algorithm.
This is akin to the framework of $K$-extendible systems by~\cite{Mestre06} in which this greedy procedure can be alternatively analyzed.
Finally, we observe that since the problem solved was a relaxation of the original ranking maximization problem, the approximation ratio we obtain with respect to the original problem is still $(\Delta+2)$. 
%%%

It remains to complete the ranking by filling in any gaps that may have been left by the above procedure. 
This can be achieved in a greedy manner that only increases the value of the solution, and violates the constraints by at most a multiplicative factor of $2$. 

\medskip
\parag{Overview of the proof of Theorem~\ref{thm:hardnessAll}.} 
Our hardness results are based on a general observation that one can encode various types of packing constraints using instances of the constrained ranking maximization (U) and feasibility (U) problem.
The first result (Theorem {\ref{thm:np-hard}})
-- $\NP$-hardness of the feasibility problem (for $\Delta \geq 3$) is established by a reduction from the hypergraph matching problem.
Given an instance of the  hypergraph matching problem one can think of its hyperedges as items and its  vertices as properties.
Degree constraints on vertices can then be encoded by upper-bound constraints on the number of items   that have a certain property in the ranking.
The inapproximability result (Theorem \ref{thm:hard_approximate})
 is also established  by a reduction from the hypergraph matching problem, however in this case  one needs to be more careful as the reduction is required to output instances that are feasible.

Our next hardness result (Theorem~\ref{thm:fixedhardness}) illustrates that the difficulty of the constrained ranking optimization problem (U) could be entirely due to the upper-bound numbers $U_{k\ell}$s.
In particular, even when  the part of the input corresponding to which item has which property is fixed, and only depends on $m$  (and, hence, can be {\em pre-processed} as in \cite{FeigeJ12}), the problem remains hard.
This is proven via a reduction from the independent set problem.
The properties consists of all pairs of items $\{i_1, i_2\}$ for $i_1, i_2 \in [m]$.
Given any graph $G=(V,E)$ on $m$ vertices, we can set up a constrained ranking problem whose solutions are independent sets in $G$ of a certain size.
Since every edge $e=\{i_1,i_2\}\in E$ is  a property, we can set a constraint  that allows at most one item (vertex) from this property (edge) in the ranking. 
  
Finally Theorem~\ref{thm:violation} states that it is not only hard to decide feasibility but even to find a solution that does not violate any constraint by more than a constant multiplicative factor $c\in \mathbb{N}.$
The obstacle in proving such a hardness result is that, typically, even if a given instance is infeasible, it is  easy to find a solution that  violates {\it many} constraints by a small amount.
To overcome this problem we employ an inapproximability result for the maximum independent set problem by \cite{Hastad96} and an idea by~\cite{CK99}.
Our reduction (roughly) puts a constraint on every $(c+1)-$clique in the input graph $G=(V,E)$, so  that at most one vertex (item) is picked from it.
Then a solution that does not violate any constraint by a multiplicative factor more than  $c$  corresponds to a set of vertices $S$ such that the induced subgraph $G[S]$ has no $c$-clique. 
Such a property  allows us to prove (using elementary bounds on Ramsey numbers) that $G$ has a large independent set.
Hence, given an algorithm that is able to find a feasible  ranking with no more than a $c$-factor violation of the constraints, we can approximate the maximum size of an independent set in a graph $G=(V,E)$ up to a factor of roughly $|V|^{1-\nfrac{1}{c}}$; which is hard by~\cite{Hastad96}.

\vspace{-2mm}
\section{A Polynomial Time algorithm for $\Delta=1$}\label{sec:flow}
\begin{proofof}{of Theorem~\ref{thm:flow}}
We show how to solve the constrained ranking maximization problem for $\Delta=1$ by reformulating it as a minimum cost flow problem in a directed network.
For clarity, we begin with the case when only upper-bound constraints are present and then modify the construction to deal with lower-bounds as well.

Given an instance of the constrained ranking maximization problem we build a directed, weighted, capacitated graph $G=(V,E)$ such that the solution of the min cost flow problem on $G$ allows us to solve the instance of the ranking problem.

There are two special vertices: $s,t\in V$ -- source and sink respectively.
For every position in the ranking $k\in [n]$ there is a vertex $\gamma_{k}\in V$ corresponding to it.
Further, for every property $\ell\in [p]$ there are $(n+1)$ vertices $\rho_{\ell,0}, \rho_{\ell,1}, \ldots, \rho_{\ell, n}$ in $V$.
Thus in total there are $|V|=(n+1)p+n+2=O(np)$ vertices in $G$.

We proceed to describing edges.
For every position $k$ there is an edge $e$ from $\gamma_k$ to $t$ with unit capacity and zero cost.
Further, for every property $\ell \in [p]$ and every position $k=1,2, \ldots, n$ there is an edge $e$ from $\rho_{\ell,k-1}$ to $\gamma_{k}$ with unit capacity and zero cost.

Finally the vertices $\{\rho_{\ell,k}\}_k$ are arranged in a directed chain from $k=n$ to $k=0$.
To describe the corresponding edges, let us denote the items having property $\ell$  (in increasing order) by $i_1, i_2, i_3,\ldots, i_{S_\ell}$ (where $S_\ell$ is the number of items having property $\ell$)
\begin{itemize}
\item  the vertex $s$ is connected to $\rho_{\ell,n}$ by $S_\ell$ edges (one for each item of this property), each of unit capacity, and their costs are zero,
\item for every $1\leq k \leq n$ there are exactly $U_{\ell,k}$ edges going from $\rho_{\ell,k}$ to $\rho_{\ell, k-1}$, each of unit capacity. Their costs are as follows
\begin{equation*}%\label{eq:edge_w}
W_{i_1,k+1} - W_{i_1,k}, ~~~W_{i_2, k+1}-W_{i_2, k}, ~~~\ldots,~~~ W_{i_{U_{\ell,k}},k+1}-W_{i_{U_{\ell,k}},k},
\end{equation*}
where we define $W_{i,n+1}:=0$ for every item $i\in [m]$.
\end{itemize}

The idea behind the above construction is that sending one unit of flow from $s$ to $\rho_{\ell, n}$ (for some position $\ell$) and then to $\rho_{\ell, n-1}$ and all the way to $\rho_{\ell, k}$ (for some position $k$) incurs a cost of (negative) $W_{i_1, k}$ and corresponds to placing $i_1$ (the ``best'' item of property $\ell$) at position $k$. The flow is then sent through $\gamma_k$ to $t$ which guarantees that no position in the ranking is used twice. 
Thus every flow of value $n$ corresponds to a feasible ranking and its cost is the negative of the total weight of that given ranking.
Similarly, the optimal ranking has a flow of value $n$ corresponding to it.
We now prove these intuitions formally.

\parag{Rankings $\to$ Flows.}
We show that if $\pi^\star: [n] \to [m]$ is the optimal ranking then it can be translated into a feasible flow of value $n$ and cost $-\sum_{k=1}^n W_{\pi^\star(k),k}$.
Let us start by a simple observation: for every property $\ell$ the relative order of appearance in $\pi^\star$ of items having this property is consistent with the ordering of their indices (i.e. $i_1<i_2<i_3<\ldots$).
This follows from the Monge property, as if this was not satisfied, then swapping some of the items (an incostistent pair) would cause the ranking to have higher total weight (see the proof of Theorem~\ref{thm:dp_formal} for more details).
 
To construct a feasible flow from $\pi^\star$, we scan the ranking from top $k=1$ to bottom $k=n$ and add one flow path at a time.
Suppose that $\pi^\star(k)=i$, i.e., item $i$ is placed at position $k$ in $\pi$.
Let $\ell$ be the unique property item $i$ has.
We send one unit of flow along the following path
\begin{equation}\label{eq:flow_path}
s\to \rho_{\ell,n} \to \rho_{\ell, n-1} \to \ldots \to \rho_{\ell, k-1} \to \gamma_k \to t
\end{equation} 
We pick the edges on this path so that all of them correspond to the item $i$ -- this is possible by our observation above.
The cost contributed by this path is
$$0+(W_{i,n+1}-W_{i,n})+(W_{i,n}-W_{i,n-1})+\ldots+(W_{i,k+1}-W_{i,k})+0+0=-W_{i,k}.$$
Thus the total cost is $-\sum_{k=1}^n W_{\pi^\star(k),k}$ as claimed.
It remains to observe that such a flow is indeed feasible, since by the construction, the capacities along the $\{\rho_{\ell,k}\}_k$ chains correspond to upper bound constraints $\{U_{\ell,k}\}_k$.

\parag{Flows $\to$ Ranking.}
Suppose $f$ is a feasible flow of value $n$ and cost $\mathrm{cost}(f)$.
We show that there is a feasible ranking $\pi$ of weight at least $-\mathrm{cost}(f)$.
Before we attempt constructing a ranking let us first modify the flow $f$ to force a certain structure on it.
Consider any property $\ell$ and any position $k$ in the ranking.
If there are $c$ units of flow sent between $\rho_{\ell, k+1}$ and $\rho_{\ell,k}$ then we will assume that they go through $c$ edges of smallest cost.
In other words, if (as before) $i_1, i_2, i_3, \ldots$ are all the items having property $\ell$ (in increasing order) then the edges we pick have costs
\begin{equation*}
%\label{eq:edge_w}
W_{i_1,k+1} - W_{i_1,k}, ~~~W_{i_2, k+1}-W_{i_2, k}, ~~~\ldots,~~~ W_{i_c,k+1}-W_{i_c,k}.
\end{equation*}
The fact that these are the edges with smallest cost follows from the Monge property.

Given such a structured flow we can easily construct a ranking with the same weight.
Suppose that there are $c$ units of flow entering the chain $\{\rho_{\ell,k}\}_k$ and leaving it through vertices $\gamma_{k_1}, \gamma_{k_2}, \ldots, \gamma_{k_c}$, then we set
$$\pi(k_1)=i_1, ~~~\pi(k_2)=i_2, ~~~\ldots, \pi(k_c)=i_c,$$
where, as above, $i_1, i_2, \ldots, i_c$ are the ``best'' items of property $\ell$.
The cost of this part of the flow (corresponding to $\ell$) is then
$$-(W_{i_1,k_1}+W_{i_2,k_2}+\ldots+W_{i_c,k_c})$$
and the total is $-\sum_{k=1}^nW_{\pi(k),k}$.
The feasibility of this ranking follows from the fact that at most $U_{\ell,k}$ units of flow reach $\rho_{\ell,k}$ and thus at most that many items of this property are ranked at top-$k$ positions.

\begin{figure}[t!]
\centering
\caption{An illustration of the graph $G$: the chain corresponding to a property $\ell$ and how it is connected to the $\gamma$-vertices. Note that the number of edges between $\rho_{\ell,k}$ and $\rho_{\ell, k+1}$ is $U_{\ell,k}$ and hence it might grow with $k$.}\label{fig:flow}
  \includegraphics[width=.8\columnwidth]{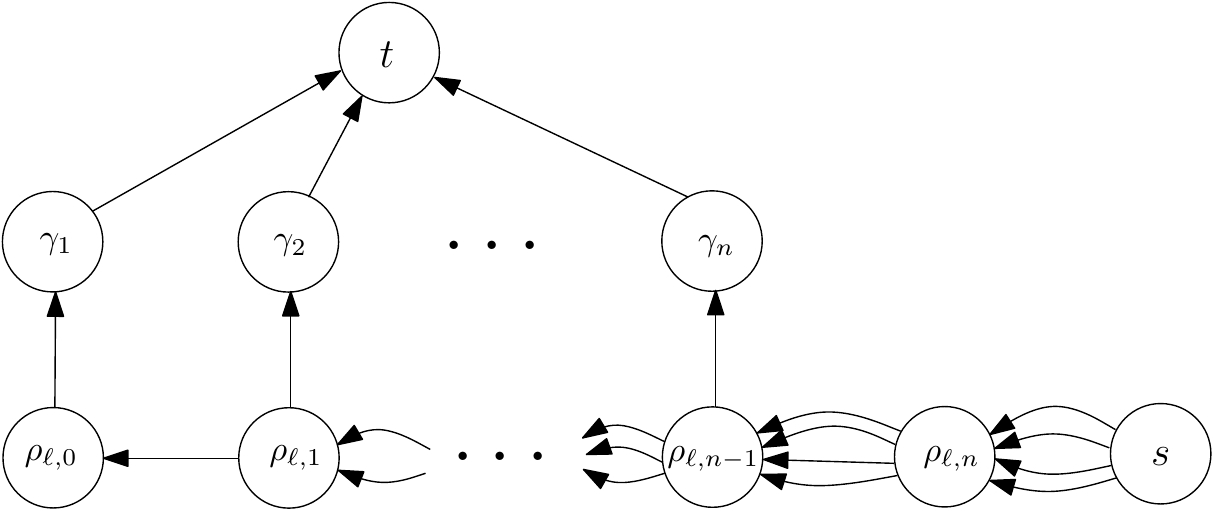}
\end{figure}

 \medskip
So far our reduction captures only the (U) variant of the ranking problem.
To deal with lower-bounds, we will slightly modify the edge costs of the constructed flow network.
To this end we introduce a number $M$ to be very large (i.e., significantly larger than any entry in the weight matrix $W$).
For a property $\ell$ and a position $k$ in the ranking we modify the costs of the edges between $\rho_{\ell,k}$ and $\rho_{\ell, k-1}$ to 
\begin{equation*}%\label{eq:edge_w}
W_{i_1,k+1} - W_{i_1,k}-M, ~\ldots, ~W_{i_{L_{\ell,k}}, k+1}-W_{i_{L_{\ell,k}}, k}-M,~W_{i_{L_{\ell,k}+1}, k+1}-W_{i_{L_{\ell,k}+1}, k} ~~\ldots,~ W_{i_{U_{\ell,k}},k+1}-W_{i_{U_{\ell,k}},k},
\end{equation*}
i.e., the $L_{\ell, k}$ smallest cost edges have their cost further decreased by $M$.
We can now make $M$ so large that we force all these $L_{\ell,k}$ edges to be sent flow through.
The total (negative) cost of this is $ \sum_{\ell \in [p]}\sum_{k\in [n]} L_{\ell,k}$.
Thus in particular, the minimum cost flow $f^\star$ still corresponds to a feasible ranking (if such exists) and the weight of the optimal ranking is recovered as
$M \cdot ( \sum_{\ell \in [p]}\sum_{k\in [n]} L_{\ell,k} )- \mathrm{cost}(f^\star).$

To solve the constructed instance of minimum cost flow, note first that the network $G$ does not contain any negative-cost cycle (in fact it is acyclic), which is a sufficient condition for the problem to be polynomial-time solvable.

One can use the Successive Shortest (augmenting) Paths algorithm (see e.g.~\cite{AMO93}) to solve this problem efficiently.
For the initial computation of the node potentials one does not need to use the Bellman-Ford algorithm (which has running time $O(|V|\cdot |E|)=O(pmn^2)$) but take advantage of the fact that $G$ is acyclic and compute single source shortest paths from $s$ by sorting the vertices topologically and running a simple dynamic programming procedure.
After that, the algorithm needs to augment the flow $n$ times, each augmentation requires to run Dijkstra's algorithm once, i.e., takes $O(|E| \log |V|)=O(mn \log m)$ time.
Thus the total running time of the algorithm is $O(nm+n\cdot mn \log m)=O(n^2m \log m)$.
\end{proofof}

\section{Dynamic Programming-based Exact Algorithm}\label{sec:dp}

Recall that for an instance of constrained ranking maximization, every item $i \in [m]$ has a type $T_i$ assigned to it, which is the set of all properties item $i$ has. 
In this section, we present an exact dynamic programming algorithm for solving the constrained ranking maximization problem which is efficient when the number of distinct types $q$ in the instance is small. 
We start by providing a geometric viewpoint of the problem, which (arguably) makes it easier to visualize and provides us with convenient notation under which the dynamic programming algorithm is simpler to state and understand.
\subsection{Geometric interpretation of fairness constraints}
\label{sec:geometric}
Recall that in an instance of constrained ranking maximization we are given $m$ items, $n$ ranking positions and $p$ properties, together with fairness constraints on them. 
Let $t_i=1_{T_i} \in \{0,1\}^p$ be the vector indicating which sets of $P_\ell$ item $i$ belongs to (we call this the \emph{type} of $i$).

Note that every ranking can be described either by a binary matrix $x\in \{0,1\}^{m \times n}$ such that $x_{ij}=1$ if and only if item $i$ is ranked at position $j$, or alternatively by a one-to-one function $\pi: [n] \to [m]$ such that $\pi(j)$ is the item ranked at position $j$, for every $j\in [n]$. 
Using the latter convention we can encode the fairness condition as
\[\forall ~{k\in [n]}~~~~L_k \leq \sum_{j=1}^k t_{\pi(j)} \leq U_k,\]
where $L_k=(L_{k1}, L_{k2}, \ldots, L_{kp})^\top$ and  $U_k=(U_{k1}, U_{k2}, \ldots, U_{kp})^\top$ are the vectors of lower and upper-bounds at position $k$. 
In other words, a ranking is feasible if and only if the $k$th partial sum of all $t_i$ vectors of items at top-$k$ positions belongs to the rectangle $[L_{k1},U_{k1}] \times [L_{k2},U_{k2}] \times \ldots \times [L_{kp},U_{kp}]$, for every $k\in [n]$.

%%%%%%%%%%
\subsection{The dynamic programming algorithm}

\begin{theorem}\label{thm:dp_formal}
There is an algorithm that solves the constrained ranking maximization problem (LU) when the objective function $w$ satisfies property \eqref{eq:monge} in $O(mp+pqn^q)$ time (where $q$ is the number of different types of items). 
\end{theorem}
\begin{proof} 
It is convenient to assume that the matrix $w$ satisfies a strict variant of property~\eqref{eq:monge} in which all the inequalities are strict. 
The general case follows by an analogous argument.
For an item $i\in [m]$, recall that $t_i=1_{T_i} \in \{0,1\}^p$ is the vector indicating which sets $P_\ell$ item $i$ belongs to. 
Further, let $Q=\{t_i: i\in [m]\}$ be the set of all realized types. 
Denote the elements of $Q$ by $v_1, v_2, \ldots, v_q\in \{0,1\}^d$. 
For every $\ell \in [q]$ define $Q_\ell = \{i \in [m]: t_i = v_\ell\}$ and let $q_\ell = |Q_\ell|$.  

For every $\ell \in [q]$ we denote by $i^{(\ell)}_1, i^{(\ell)}_2, \ldots, i^{(\ell)}_{q_\ell}$ the list of items in $Q_\ell$ in increasing order. 
Note that if in an optimal solution to the ranking maximization problem, exactly $s_\ell$ items come from $Q_\ell$, then these items are exactly $i^{(\ell)}_{1}, \ldots, i^{(\ell)}_{s_\ell}$ and they appear in increasing order in the solution. 
This follows from property \eqref{eq:monge} of $w$ as follows: 
Suppose that an item $i_1\in Q_\ell$ is placed at position $j_2$ and an item $i_2 \in Q_\ell$ is placed at position $j_1$, with $i_1<i_2$ and $j_1<j_2$. 
Swapping these two items in the ranking does not affect feasibility of the solution and the difference in value is
\[W_{i_1j_1}+W_{i_2j_2}-W_{i_1j_2}-W_{i_2j_1}.\]
This is positive due to the (strict) property \eqref{eq:monge}. 
Hence the swap can only increase the weight of the solution. 
A similar reasoning shows that it is beneficial to swap a ranked item $i_2$ with an unranked item $i_1$, whenever $i_1<i_2$.

One of the consequences of the above observations is that we can assume that $q_\ell \leq n$ for all $\ell \in [q]$ and hence $m \leq nq$. 
This is because we can keep at most $n$ best items from every set $Q_\ell$ and discard the remaining ones as they will not be part of any optimal solution. 
Such discarding can be done in time roughly $O(m)$ if an instance with $q_\ell > n$ is given. 

The above allows us to reduce the number of candidate rankings which one has to check to roughly $q^n$. However, this number is still prohibitively large.
As $q \ll n$ in many scenarios of interest, we construct a dynamic programming algorithm with a much fewer states $O(n^q)$. 

Now, consider the following sub-problem: 
For any tuple $(s_1, s_2, \ldots, s_q)\in \N^q$ with $k=\sum_{\ell=1}^q s_\ell \leq n$ let $D[s_1, s_2, \ldots, s_q]$  be the largest weight of a feasible ranking with top-$k$ positions occupied, such that exactly $s_\ell$ items are picked from $Q_\ell$ for every $\ell \in [q]$. 
Let us now describe an algorithm for computing $D[s_1, s_2, \ldots, s_q]$.
First, initialize all entries $D[s_1, s_2, \ldots, s_q]$ to $-\infty$ and set $D[0,0,\ldots, 0]=0$. 
Next, consider all valid tuples $(s_1, s_2, \ldots, s_q)$ in order of increasing values of $k=\sum_{\ell=1}^q s_\ell$, i.e., $k=1,2, \ldots, n$. 
Suppose that we would like to compute $D[s_1, s_2, \ldots, s_q]$. 
First one must check whether the fairness constraint at position $k$ is satisfied; for this we calculate
\[v=\sum_{\ell=1}^q s_\ell v_\ell.\]
Note that the $\ell$th coordinate of  $v$ represents the number of items having property $\ell \in [p]$. 
Hence, a necessary condition for the tuple $(s_1, s_2, \ldots, s_q)$ to represent a feasible ranking is that $L_k \leq v \leq U_k$. 
If that is not satisfied we just set $D[s_1, s_2, \ldots, s_q]=-\infty$. 
Otherwise, consider all possibilities for the type of item that is placed at the last position: $k$. 
Suppose it is of type $\ell$ (i.e., it belongs to $Q_\ell$). 
Then we have
\[D[s_1, s_2, \ldots, s_q] = D[s_1, \ldots,s_{\ell-1}, s_\ell -1, s_{\ell+1}, \ldots,  s_q] + W_{ik}\]
where $i=i^{(\ell)}_{s_\ell}.$
Hence, in order to compute $D[s_1, s_2, \ldots, s_q]$ we simply iterate over all possible types $\ell \in [q]$ and find the maximum value we can get from the above. 
Correctness follows from the fact that the $s_\ell$th item of type $Q_\ell$ in every optimal ranking is always $i^{(\ell)}_{s_\ell}$. 

The total number of sub-problems is at most $O(n^q)$. 
Hence, the above algorithm can be implemented in time  $O(m \cdot p+p\cdot q \cdot n^q)$, where $m\cdot p$ time is required to read the input and construct the list of elements of every given type. The second term $O(p \cdot q \cdot n^q)$ appears because there are $n^q$ subproblems, each such sub-problem considers $q$ cases, and every case has a feasibility check that takes $O(p)$ time.
\end{proof}

\section{Algorithms for $\Delta=1$ and Upper-Bound Constraints}\label{sec:lp}

\subsection{Integrality of LP solutions}
\label{sec:lpint}

\begin{theorem}\label{thm:lp1_formal}
Consider the LP relaxation~\eqref{eq:lp} for the constrained ranking maximization problem (U) when the properties $P_1, \ldots, P_p$ are pairwise disjoint (i.e., $\Delta=1$). 
If $W\in \R^{m \times n}_{\geq 0}$  satisfies \eqref{eq:monge}, then there exists an optimal integral solution $x^\star \in \{0,1\}^{m \times n}$ to~\eqref{eq:lp}.
\end{theorem}

\begin{proof}
Without loss of generality, we can assume that $n=m$ via a simple extension of the problem as follows: 
Extend the matrix $w\in \R^{m \times n}$ to a square matrix $\tilde{w} \in \R^{m \times m}$ by setting 
\[ \tilde{w}_{ij} =
\begin{cases} 
W_{ij} & \mbox{ for $i \in [m]$ and $j\in [n]$}\\
0 & \mbox{ for $i\in [m]$ and $j\in \{n+1, \ldots, m\}$}.
\end{cases}\]
Further, for every $\ell\in [p]$ and $k\in \{n+1,\ldots,m\}$ we set $U_{k\ell}=k$; i.e.,  no constraint is imposed on these positions. 
Note that $\tilde{w}$ still satisfies property \eqref{eq:monge}. 
Moreover, every solution $x\in \Omega_{m,n}$ to the original problem~\eqref{eq:lp} can be extended to a solution $\tilde{x} \in \Omega_{m,m}$ while preserving the weight; i.e., $\inangle{w,x}= \inangle{\tilde{w}, \tilde{x}}$ (where $\inangle{\cdot, \cdot}$ denotes the inner product between two matrices, i.e., $\inangle{w,x} = \sum_{i=1}^m \sum_{j=1}^m W_{ij}x_{ij}$). 
Similarly, every solution $\tilde{x}$ to the extended problem, when restricted to first $n$ columns, yields a solution to the original problem with the same weight. 
Thus, it suffices to prove that the extended problem has an optimal integral solution, and for the remainder of this section we assume that $n=m$.

For simplicity, assume that the matrix $w$ satisfies the {\it strict} variant of    property \eqref{eq:monge}; i.e., when the inequalities in~\eqref{eq:monge} are strict.
This can be achieved by a small perturbation of the weights without changing the optimal ranking.

Our proof consist of two phases. In the first phase, we show that every optimal solution satisfies a certain property on its support. 
In the second phase we show that no optimal solution that has this property can have  fractional entries.  
Let us state the property of a feasible solution $x\in \Omega_{m,m}$ that we would like to establish: 
\begin{equation}\label{eq:cond}
\begin{aligned}
&\forall_{\ell \in [p]} ~~\forall_{i_1, i_2\in P_\ell}~~ \forall_ {j_1, j_2\in [m]} ~~~~\\
(i_1<i_2 &\wedge j_1<j_2) ~~\Rightarrow~~  x_{i_2j_1}\cdot x_{i_1j_2}=0.
\end{aligned}
\end{equation}
In other words, whenever we have two items $i_1, i_2$ that have the same property $\ell$,  if $i_1$ is before $i_2$ (i.e., $i_1$ is better than $i_2$) then for any position $j_1, j_2 \in [m]$ such that $j_1$ is above $j_2$, then $x_{i_2 j_1}$ and $x_{i_1 j_2}$ cannot both be positive.  
We show that if $x$ does not satisfy condition \eqref{eq:cond} then it is not optimal. 

To this end, take a fractional solution $x$ for which there is some $i_1, i_2 \in P_\ell$ and $j_1, j_2\in [m]$ for which the condition does not hold. 
Now, consider a solution of the form $x^\prime=x+ y$ where $y = \eps (e^{(i_1 ,j_1)}+e^{(i_2 ,j_2)} - e^{(i_1 ,j_2)} - e^{(i_2 ,j_1)})$ for some $\eps>0$ and $e^{(i,j)}\in \R^{m\times m}$ denotes the matrix with a single non-zero entry at $(i,j)$ of value 1.
  Since $x_{i_1j_2}>0$ and $x_{i_2j_1}>0$, we can find some $\eps>0$ such that $x^\prime\geq 0$. 
 Furthermore, we claim that such a solution $x^\prime$ is still feasible for~\eqref{eq:lp}.
  Indeed, for every item $i\in [m]$ we have $\sum_{j=1}^n y_{ij}=0$ we can conclude that 
\[\sum_{j=1}^m x^\prime_{ij} = \sum_{j=1}^m x_{ij}\leq 1.\]
Similarly, for every rank position $j\in [m]$, we have $\sum_{i=1}^m x^\prime_{i,j}=\sum_{i=1}^m x_{i,j}=1$.
Hence, $x^\prime\in \Omega_{m,m}$. 

It remains to show that $x^\prime$ satisfies all of the fairness constraints. 
Note that it is enough to consider fairness constraints coming from the property $\ell \in [p]$, as $i_1, i_2\in P_\ell$,  $i_1, i_2\not\in P_{\ell^\prime}$ for any $\ell^\prime \neq \ell$, and variables $x_{i_1 j_1}, x_{i_1 j_2}, x_{i_2 j_1}, x_{i_2j_2}$ do not appear in other constraints. 
Every such constraint is of the form $\inangle{1_{R_k}, x} \leq U_{k\ell}$ where $1_{R_k}$ is the indicator vector (matrix) of the rectangle (i.e., submatrix) $P_\ell \times [k]$. 
Since $\inangle{1_{R_k}, y}=0$ for every such rectangle, we have
$$\inangle{1_{R_k}, x^\prime} = \inangle{1_{R_k}, x+\eps y} =\inangle{1_{R_k},x} \leq U_{k\ell}.$$
Therefore $x^\prime$ is feasible for~\eqref{eq:lp}. 
Furthermore, because of the (strict) property \eqref{eq:monge}, we have:
\[\inangle{w,x^\prime} = \inangle{w,x} + \eps(W_{i_1j_1}+W_{i_2j_2}-W_{i_1j_2}-W_{i_2j_1})> \inangle{w,x}.\]
As this is a feasible solution with a strictly better objective value, we conclude that $x$ was not optimal. 
Hence, every optimal solution necessarily satisfies~\eqref{eq:lp}.

Suppose now, for sake of contradiction, that $x$ satisfies~\eqref{eq:lp} and $x$ is not integral.
Consider a fractional entry $x_{i_0j_0}\in (0,1)$ of $x$ with $i_0$ as small as possible, and (in case of a tie) $j_0$ as small as possible. 
Suppose that the item $i_0$ belongs to $P_\ell$ for some $\ell \in [p]$. 
Note that there exists an entry $(i_0,j_1)$ with $j_1>j_0$ such that $x_{i_0j_1}>0$. 
This is due to the fact that 
\[\sum_{j=1}^m x_{i_0j}=1~~~\mbox{ and }~~~~\sum_{j=1}^{j_0} x_{i_0 j}=x_{i_0j_0}<1.\]
Fix the smallest possible $j_1$ with this property. 
Because of the constraint $\sum_{i=1}^m x_{i j_0}=1$, there exists at least one more fractional entry in the $j_0$th column, let us call it $x_{i_1 j_0}$. 
It follows that $i_1>i_0$. 
Note also that $i_1 \notin P_\ell$, as if $i_1 \in P_\ell$ then condition~\eqref{eq:cond} would be violated. 

Let us again consider a new candidate solution using the indices defined above: 
\[x^\prime:=x+\eps(e^{(i_1 ,j_1)}+e^{(i_2 ,j_2)} - e^{(i_1 ,j_2)} - e^{(i_2 ,j_1)}).\]
We show that $x^\prime$ is feasible for some $\eps>0$, which then contradicts the fact that $x$ is optimal because of the strict version of property \eqref{eq:monge}. 
To do this, it suffices to ensure that $x^\prime$ does not violate any fairness constraints imposed by the property $\ell \in [p]$. 
Note that for $k \notin \{ j_0, j_0+1,\ldots, j_1\}$ the constraints 
\[\sum_{i\in P_\ell} \sum_{j=1}^k x^\prime_{ij} =\sum_{i\in P_\ell} \sum_{j=1}^k x_{ij} \leq U_{k\ell}\]
remain satisfied.
Hence, it only remains to show that no constraint $\sum_{i\in P_\ell} \sum_{j=1}^k x_{ij} \leq U_{k\ell}$ is tight at $x$ for $k\in \{ j_0, j_0+1,\ldots, j_1\}$.

Observe that because of our choice of $(i_0, j_0)$, all entries in the rectangle $P_\ell \times [j_0-1]$ are integral. 
Furthermore, in the rectangle $P_\ell \times \{j_0, j_0+1, \ldots, j_1-1\}$, the only non-zero entry is $x_{i_0,j_0}\in (0,1)$ due to the fact that $x_{i_0,j_1}>0$ and condition~\eqref{eq:cond} is satisfied. 
Now, because $U_{k\ell} \in \mathbb Z$ but for $k$ as above $\sum_{i\in P_\ell} \sum_{j=1}^k x_{ij} \not\in \mathbb Z$, the constraint cannot be tight.
Thus, $x^\prime$ is feasible for some $\eps>0$ and hence $x$ is not an optimal solution to~\eqref{eq:lp}. 
Hence, no optimal solution has fractional entries. 
\end{proof}

In contrast to the above theorem, some vertices of the feasible region might be non-integral.
\begin{fact}\label{non_integral}
There exists an instance of the ranking maximization problem for $\Delta=1$, such that the feasible region of~\eqref{eq:lp} has fractional vertices.
\end{fact}

\begin{proof}Let $n=m=4$ and suppose there is only one property $P_1=\{1,2\}$ and the constraints are $U_{21}=1$ and $U_{k1}=\infty$ for $k \neq 2$. In other words, we only constrain the ranking to have at most $1$ element of property $1$ in the top-2 entries. 

Consider the following point.
$$x=
\begin{pmatrix}
\nfrac12 & 0 & 0 & \nfrac12 \\
0 & \nfrac12 & \nfrac12 & 0\\
\nfrac12 & 0 & \nfrac12 & 0 \\
0 & \nfrac12 & 0 & \nfrac12
\end{pmatrix}$$
Clearly $x$ is feasible. Observe that the support of $x$ has $2n$ elements and there are exactly that many linearly independent tight constraints at that point.  Indeed the doubly-stochastic constraints give us $2n$ constraints out of which $2n-1$ are linearly independent, and the remaining one is 
$$x_{1,1}+x_{1,2}+x_{2,1}+x_{2,2} = 1.$$
Therefore $x$ is a (non-integral) vertex of the feasible region of~\eqref{eq:lp}.
\end{proof}

%%%%%%%%%%%%%%
\subsection{Fast greedy algorithm}\label{sec:greedy}

Due to the special structure for $\Delta = 1$, we are able to find a fast simple algorithm for the ranking maximization problem in this case. 
\begin{theorem}\label{thm:FastGreedy}
There exists an algorithm which, given an instance of the constrained ranking maximization problem (U) with $\Delta=1$ and objective function that satisfies property \eqref{eq:monge}, outputs an optimal ranking in $O(m+n \cdot p)$ time.
\end{theorem}

\begin{proof}
For simplicity, assume that $w$ satisfies the strict variant of property \eqref{eq:monge} (with strict inequalities in the definition). 
This can be assumed without loss of generality by slightly perturbing $w$.
Consider the following greedy algorithm that iteratively constructs a ranking $\pi: [n] \to [m]$ (i.e., $\pi(j)$ is the item ranked at position $j$, for all $j\in [n]$).\footnote{This alternate notation makes the exposition in this section cleaner -- see also the notation and problem formulation in Section~\ref{sec:geometric}.}
\begin{itemize}
%[topsep=1pt,itemsep=0ex,partopsep=0ex,parsep=1pt]
\item For $j=1$ to $n$
\begin{itemize}
%[topsep=1pt,itemsep=0ex,partopsep=0ex,parsep=1pt]
\item Let $i\in [m]$ be the smallest index of an item which was not yet picked and can be added at position $j$ without violating any constraint. If there is no such $i$, output INFEASIBLE.
\item Set $\pi(j)=i$.
\end{itemize}
\item Output $\pi$.
\end{itemize}
\noindent It is clear that if the above algorithm outputs a ranking $\pi: [n] \to [m]$ then $\pi$ is feasible. Assume now that it indeed outputs a ranking. We will show that it is optimal.

Take any optimal ranking $\pi^\star$. 
Let $P_\ell$ be any property (for $\ell \in [p]$) and let  $i_1, i_2, \ldots, i_{s}$ be the list of items in $P_{\ell}$ in increasing order. 
We claim that if $\pi^\star$ ranks exactly $r$ items from $P_\ell$ then these have to be $i_1, i_2, \ldots, i_r$, in that order. 
For this, note that when swapping two elements, say $i_1, i_2 \in P_\ell$, at positions $j_2,j_1$ in the ranking (with say $j_1<j_2$) the change in weight is equal to
\[W_{i_1j_1}+W_{i_2j_2}-W_{i_1j_2}-W_{i_2j_1}>0\]
because of the (strict)  property \eqref{eq:monge}. 
Hence it is always beneficial to rank the items in $P_\ell$ in increasing order. 
Furthermore, it can be argued using monotonicity that it is always optimal to select the $r$ items with smallest indices for the ranking. 

One of the consequences of the above observations is that we can assume that $|P_\ell | \leq n$ for all $\ell \in [q]$ and hence $m \leq np$. 
This is because we can keep at most $n$ best items with property $\ell$, and discard the remaining ones as they will not be part of any optimal solution. 
Such a discarding can be done in time $O(m)$ time if an instance with $|P_\ell | > n$ is given. 

Further, the above observation allows us to now prove optimality of the greedy strategy. 
Take the largest number $k$ such that $\pi$ and $\pi^\star$ agree on $[k-1]$, i.e., $\pi(j) = \pi^\star(j)$ for $j=1,2 ,\ldots, k-1$. If $k-1=n$ then there is nothing to prove. Let us then assume that $k-1<n$ and $\pi(k)=i_A \neq i_O = \pi^\star(k)$. 
There are two cases: either $i_A$ is ranked in $\pi^\star$ or it is not. 

In the first case, let $k^\prime$ be the position in $\pi^\star$ such that $\pi^\star(k^\prime) = i_A$, clearly $k^\prime>k$. Let $\hat{\pi}$ be a ranking identical to $\pi^\star$ but with positions $k$ and $k^\prime$ swapped. We claim that $\hat{\pi}$ is still feasible and has larger weight than $\pi^\star$. The claim about weights follows easily from the strict  property \eqref{eq:monge}. Let us now reason about feasibility of $\hat{\pi}$. Let $\ell$ be the only property of $i_A$ (i.e. $i_A \in P_\ell$) and let $h$ be the total number of elements of property $\ell$ in top-$(k-1)$ positions of $\pi$ (or equivalently of $\pi^\star$). Note that by doing the swap we could have only violated some constraint corresponding to $\ell$. Since $\pi(k) = i_A$ we know that $U_{\ell k}\geq h+1$ (and similarly $U_{k+1,\ell}, U_{k+2,\ell}, \ldots, U_{n,\ell}\geq h+1$).  
Further, because of our previous observation, no item $i\in P_\ell$ is ranked at position $j$ for $k\leq j<k^\prime$ in $\pi^\star$. For this reason, the fairness constraints corresponding to $\ell$ at $j=k,k+1, \ldots, k^\prime$ are satisfied for $\hat{\pi}$ (there are $h+1$ elements of property $\ell$ in top-$k^\prime$ items in $\hat{\pi}$).  
The second case is very similar. One can reason that if $i_A$ is not included in the ranking $\pi^\star$ then by changing its $k$th position to $i_A$ we obtain a ranking which is still feasible but has larger value.

Hence, if $k-1<n$, this contradicts the optimality of $\pi^\star$. Thus, $\pi=\pi^\star$ and $\pi$ is the optimal ranking. 
By the same argument, one can show that if the instance is feasible, the greedy algorithm will never fail to output a solution (i.e., report infeasibility).

Let us now discuss briefly the running of such a greedy algorithm. For every property $\ell$ we can maintain an ordered list $L_\ell$ of elements of $P_\ell$ which were not yet picked to the solution and a count $C_\ell$ of items of property $\ell$ which are already part of the solution. Then for every ranking position $k\in [n]$ we just need to look at the first element of every list $L_\ell$ for $\ell \in [p]$ and one of them will be ``the best feasible item''. Having the counters $C_\ell$ we can check feasibility in $O(1)$ time and  we can also update our lists and counters in $O(1)$ per rank position. For this reason, every rank position is handled in $O(p)$ time.  Note also that at the beginning all the lists can be constructed in total $O(m)$ time, since we can go over the items in the reverse order and place every item at the beginning of a suitable list $L_\ell$ in $O(1)$ time. Hence, the total running time is $O(m+n\cdot p)$. 
\end{proof}

\section{A $(\Delta+2)$-Approximation Algorithm}\label{sec:mul} 
\begin{theorem}\label{thm:mul_formal}
There exists a linear time algorithm which given an instance of the constrained
 ranking maximization problem (U) satisfying condition~\eqref{eq:feas1}, outputs a ranking $x \in \{0,1\}^{m \times n}$ whose weight is at least $\frac{1}{ \Delta+2}$ times the optimal one and satisfies all fairness constraints up to a factor of $2$, i.e., 
 \[\sum_{i\in P_{\ell} }\sum_{j=1}^k x_{i,j} \leq 2U_{k\ell}, ~~\mbox{for all $\ell\in [p]$ and $k\in [n]$}.\]
\end{theorem}
\begin{proof}
The algorithm can be divided into two phases: First we construct a partial ranking that may leave some positions empty, and then we refine it to yield a complete ranking. 

The first phase is finding a (close to optimal) solution to the following integer program: 
\begin{equation}
\begin{aligned}
	\max_{x\in \rOmega_{m,n}}  ~~ & \sum_{1=1}^m \sum_{j=1}^m  x_{ij} W_{ij} 
		~~~~ \\
		\mathrm{s.t.} ~~ &\sum_{i\in P_\ell}\sum_{j=1}^k x_{ij}\leq U_{k\ell}, &\mbox{ for all $\ell\in [p]$ and $k\in [n]$}
\end{aligned}
\label{eq:lp_re}	
\end{equation}
where $\rOmega_{m,n}$ is the set of all $x\in \{0,1\}^{m \times n}$ such that for all $i\in [m]$ and $j\in [n]$ 
\[ \sum_{j=1}^n x_{ij} \leq 1~~~~ \mbox{ and } ~~~~\sum_{i=1}^m x_{ij}\leq 1 .\]
Note that this is a relaxation of our problem as it does not require every position in the ranking to be filled. 
In the proof we will often refer to a pair $(i,j) \in [m] \times [n]$ as a {\it cell}, since we treat them as entries in a 2-dimensional array. 
Consider the following greedy algorithm which can be used to find a (likely non-optimal) integral solution $z\in \{0,1\}^{m \times n}$ to the problem~\eqref{eq:lp_re}.
\begin{itemize}
%[topsep=1pt,itemsep=0ex,partopsep=0ex,parsep=1pt]
%
\item Order the cells $(i,j)$ according to non-increasing values of $W_{ij}$.
\item Set $z_{ij}=0$ for all $(i,j) \in [m]\times [n]$. 
\item Process the cells $(i,j)$ one by one:
\begin{itemize}
%[topsep=1pt,itemsep=0ex,partopsep=0ex,parsep=1pt]
%
\item if adding $(i,j)$ to the solution does not cause any constraint violation, set $z_{ij}=1$,
\item otherwise, move to the next cell.
\end{itemize}
\end{itemize}
We claim that the solution $z\in \{0,1\}^{m\times n}$ given by this algorithm has value at least $\frac{1}{\Delta+2}$ times the optimal solution.
In other words, 
\begin{equation}\label{eq:phase1_opt}
\sum_{i=1}^m \sum_{j=1}^m  z_{ij} W_{ij} \geq (\Delta+2)^{-1} \OPT,
\end{equation}
where $\OPT$ denotes the optimal solution to~\eqref{eq:lp_re}.
To prove it, let us denote by 
$$S:=\{(i_1,j_1), (i_2, j_2), \ldots, (i_N, j_N)\}$$ the set of cells picked by the algorithm (in the same order as they were added to the solution).
Let also $S_r:=\{(i_1,j_1), (i_2, j_2), \ldots, (i_r, j_r)\}$ for $r=0,1, \ldots, N$.
For $r=0,1, \ldots, N$ denote by $\OPT_r$ the optimal solution to~\eqref{eq:lp_re} under the condition that the cells from $S_r$ are part of the solution, or in other words, under the additional constraint that $x_{ij}=1$ for all $(i,j)\in S_r$. We claim that for every $r=1, \ldots, N$ it holds that

\begin{equation}\label{eq:claim_greedy}
\OPT_{r-1}\leq \OPT_{r}+(\Delta+1)W_{i_rj_r}
\end{equation}
Note that given~\eqref{eq:claim_greedy}, the claim~\eqref{eq:phase1_opt} follows easily. 
Indeed, we have $\OPT=\OPT_0$ and
\begin{align*}
\OPT_0&\leq \OPT_1 + (\Delta+1)W_{i_1j_1} \\
&\leq \OPT_2 + (\Delta+1)(W_{i_1j_1}+W_{i_2j_2})\\
&\leq \ldots \\
&\leq \OPT_N + (\Delta+1)\sum_{r=1}^N W_{i_rj_r} \\
&=(\Delta+2)\sum_{r=1}^N W_{i_rj_r}
\end{align*}
The last equality follows from the fact that $S_N$ is a maximal solution (it cannot be extended), hence indeed $$\OPT_N =\sum_{i=1}^m \sum_{j=1}^n z_{ij}W_{ij}= \sum_{r=1}^N W_{i_rj_r}.$$
Thus it remains to prove~\eqref{eq:claim_greedy}.
For this, take any $r\in \{1,2, \ldots, N\}$ and consider the optimal solution $S'$ (a set of cells) which contains $S_{r-1}$, thus $\sum_{(i,j)\in S'} W_{ij}=\OPT_{r-1}$. 
If $(i_r,j_r)\in S'$ then we are done, as~\eqref{eq:claim_greedy} follows straightforwardly.
Otherwise consider $S'':=S' \cup  \{(i_r, j_r)\}$. We claim that there is a set of cells $R \subseteq S' \setminus S_{r-1}$ with $|R| \leq \Delta+2$ such that $S'' \setminus R$ is feasible for the relaxed problem~\eqref{eq:lp_re}. This then implies~\eqref{eq:claim_greedy}, as $S_r \subseteq S'' \setminus R$ and every cell $(i,j)\in R$ has weight $W_{ij}\leq W_{i_rj_r}$. 

To construct $R$ pick from $S'$ the cell lying in the same column as $(i_r, j_r)$ and the cell lying in the same row as $(i_r, j_r)$ (if they exist) and add to $R$. This is to make sure that in $S'' \setminus R$ there is at most one cell in every row and every column.
Now consider any property $\ell$, for which fairness constraints might have been violated in $S''$ (then necessarily $\ell \in T_{i_r}$).
Moreover, because of the (laminar) matroid structure of the constraints corresponding to a single property $\ell$, one can find a cell $(i,j)\in S' \setminus S_{r-1}$ such that $S'' \setminus \{(i,j)\}$ satisfies all fairness constraints for property $\ell$. Thus, by adding to $R$ one such cell (whenever it exists) for every property $\ell\in T_{i_r}$ we obtain a set $R$ of at most $\Delta+2$ cells such that $S'' \setminus R$ is feasible.
This concludes the proof that the above algorithm finds an  $(\Delta+2)$-approximation to the integer program~\eqref{eq:lp_re}.

Let us now show how to construct a full ranking out of it. 
Let 
\begin{align*}
I&:=\{i\in [m]:z_{ij}=0 \mbox{ for all }j\in [n]\}~~~~~~~~\\
J&:= \{j\in [n]:z_{ij}=0 \mbox{ for all }i\in [m]\}.
\end{align*}
We will construct a new solution $y\in \{0,1\}^{m \times n}$ such that $y$ is supported on $I \times J$,  it has at most one $1$ in every row and column, it has exactly one $1$ in every column $j\in J$ and satisfies all fairness constraints. 
Note that if we then define $x:=z+y$, then $x$ has the properties as claimed in the theorem statement. 
Indeed $x\in P_{m,n}$ and further for every $\ell \in [p]$ and $k\in [n]$
\begin{align*}
\sum_{i\in P_\ell}\sum_{ j\leq k} x_{ij} &\leq \sum_{i\in P_\ell }\sum_{ j\leq k} (z_{ij}+y_{ij }) \\
&\leq \sum_{i\in P_\ell }\sum_{ j\leq k} x_{ij}+\sum_{i\in P_\ell }\sum_{ j\leq k} y_{ij}\\
&\leq  U_{k\ell}+U_{k\ell}.
\end{align*}
The approximation guarantee follows because $y$ has a nonnegative contribution to the total weight and the guarantee on $x$ is with respect to a relaxed problem~\eqref{eq:lp_re} whose optimal value is an upper-bound on the optimal value of the ranking maximization problem. 
Hence it remains to find $y$ with properties as above.

We can construct $y$ using another greedy procedure. 
Let $y=0$. 
We process the elements of $J$ one at a time in increasing order. 
When considering a given $j$ pick any $i\in I$ such that adding $(i,j)$ to the solution $y$ does not introduce fairness constraint violation. 
It is clear that if the above algorithm succeeds to find a suitable $i\in I$ at every step, then it succeeds to construct a solution with the required properties. 
It remains to show that this is indeed the case. 
To this end, fix $j\in J$ and look at the step in which $j$ is considered. 
From condition~\eqref{eq:feas1}, we know that there are at least $n$ elements in $[m]$ that can be placed at position $j$ without violating any constraint. 
Out of these $n$ elements, some may have been already taken by the above algorithm (i.e., they are not in $I$) and some of them could have been added to the new solution $y$. 
However, as there were $n$ of them to begin with, at least one remains and can be selected.
This concludes the proof of correctness of the above procedure.
It remains to observe that both phases of the algorithm (after simple preprocessing) can be implemented in linear time, i.e., $O((m+n)\cdot p)$.
\end{proof}

%%%%%%%%%%%%%%%%%%%%%%%%%%%
\section{Hardness Results}\label{sec:hardness}
In this section we state and prove our hardness results regarding the constrained ranking feasibility and maximization problems for the (U) variant.  

Since some of the proofs involve the hypergraph matching problem, let us define this concept formally. 
A hypergraph is a pair $H=(V,E)$ composed of a vertex set $V$ and a hyperedge set $E\subseteq 2^{[m]}$. 
Given $\Delta\in \N$ we call $H$ a $\Delta$-hypergraph if every hyperedge $e\in E$ has cardinality at most $\Delta$. 
The hypergraph matching is the following decision problem: 
given a hypergraph $H$ and a number $n\in \N$, decide whether there exists a set of $n$ pairwise disjoint hyperedges in $H$ (such a set is called a matching). 
The optimization variant of this  problem asks for a largest cardinality matching.  

The name $\Delta$ above is intended to be suggestive -- recall that in the context of the constrained ranking feasibility (or maximization) problem, the parameter $\Delta$ is the maximum number of properties any given item has. 
The first hardness result states that even for small $\Delta$ the constrained ranking feasibility is hard.
\begin{theorem}\label{thm:np-hard}
The constrained ranking feasibility problem (U)  is $\NP$-hard for any $\Delta \geq 3$.
\end{theorem}
\begin{proof}
We present a reduction from $\Delta$-hypergraph matching, which is known to be $\NP$-hard for $\Delta\geq 3$. 
Let $H=(V,E)$ be a $\Delta$-hypergraph and let $n\in \Z$ be a number for which we want to test whether there is a matching of size $n$. 
We construct an instance of the constrained ranking problem whose feasibility is equivalent to $H$ having a matching of cardinality $n$ as follows.
Let $m$ be the number of hyperedges in $H$ indexed by $e_1, e_2, \ldots, e_m$, and let $p$ be the  number of vertices in $H$ indexed by $v_1, v_2, \ldots, v_p$. 
Set the number of items in the ranking problem to be $m$ where the $i$th item corresponds to edge $e_i$, and set the number of positions in the ranking to be $n$. 
For every vertex $v_\ell \in V$, introduce a property as follows:
\[P_\ell=\{i\in [m]: v_\ell \in e_i\}.\]
Thus, there are $p = |V|$ properties. 
Note that as each hyperedge $e_i$ contains at most $\Delta$ vertices, $|T_i| \leq \Delta$ for all items $i\in [m]$ as desired. 
For every $k \in [n]$ and $\ell\in [p]$, let $U_{k\ell}=1$. 

It remains to argue that $H$ has a matching of cardinality $n$ if and only if the instance of the constrained ranking problem is feasible. 
If $H$ has a matching $M$ of cardinality $n$ then let $S:=\{i\in [m]: e_i\in M\}$. 
Define a ranking by taking the elements of $S$ in any order.
Since $M$ is a matching, every vertex $v_\ell \in V$ belongs to at most one hyperedge in $M$. 
Thus, for every property $\ell\in [p]$ we have $|P_\ell \cap S| \leq 1$, and hence all of the fairness constraints are satisfied. 
Similarly, the reasoning in the opposite direction (i.e., that a feasible ranking yields a matching of cardinality $n$) follows by letting $M$ be the set of hyperedges corresponding to the $n$ items that were ranked. 
\end{proof}

\noindent
One possibility is that it is the feasibility of this problem that makes it hard. 
However, this is not the case.
Our next Theorem says that even if we guarantee feasibility via assumption~\eqref{eq:feas1}, the problem remains hard to approximate. 
\begin{theorem}\label{thm:hard_approximate}
A feasible constrained ranking maximization problem (U) that satisfies condition~\eqref{eq:feas1} and has weights that satisfy property \eqref{eq:monge} is $\NP$-hard to approximate within a factor of $O(\nfrac{\Delta}{\log \Delta})$.
\end{theorem}

\begin{proof}
In the proof we will use the fact that the maximum $\Delta$-hypergraph matching problem is hard to approximate within a factor of $O(\nfrac{\Delta}{\log \Delta})$~\cite{HSS03}. 
We present an approximation-preserving reduction from the maximum $\Delta$-hypergraph matching problem. 
The proof is similar to that of of Theorem~\ref{thm:np-hard}, however we must now ensure that the resulting instance is feasible and satisfies condition~\eqref{eq:feas1}. 

For a hypergraph $H=(V,E)$ with $|V|=p$ vertices $v_1, v_2, \ldots, v_p$ and $m=|E|$ hyperedges $e_1, e_2, \ldots, e_m$, construct an instance of constrained ranking maximization problem as follows. 
Let $m^\prime:=2m$ be the number of items and let $n:=m$ be the number of positions in the ranking. 
The items $1,2, \ldots, m$ correspond to edges $e_1, e_2, \ldots, e_m$, and the remaining items $m+1, m+2, \ldots, 2m$ we call {\it improper} items. 
As in the proof of Theorem~\ref{thm:np-hard}, for every vertex $v_\ell\in V$ we define a property $P_\ell = \{i\in [m]: v_\ell \in e_i\}$. 
The fairness constraints are defined by upper-bounds $U_{k\ell}=1$ for all $k \in [n]$ and $\ell\in [p]$. 
Improper items do not belong to any property, and hence condition~\eqref{eq:feas1} is satisfied. Furthermore, they can never cause a fairness constraint to be violated. 
Lastly, we let 
\[W_{ij} = \begin{cases}
1 & ~~\mbox{if $i\in [m]$},\\
0 &~~\mbox{if $i\in [m+1,2m].$}
\end{cases}\]
Note that such a matrix satisfies the property \eqref{eq:monge}. 

Every matching $M \subseteq E$ of cardinality $k$ in $H$ corresponds to a feasible ranking of total weight $k$ in our instance of ranking maximization (simply take the $k$ items corresponding to edges in the matching and add any $(m-k)$ improper items). 
Similarly, every ranking of weight $k$ can be transformed into a matching of size $k$. 
As the reduction is approximation preserving, we conclude that the constrained ranking maximization problem is $\NP$-hard to approximate within a factor better than $O(\nfrac{\Delta}{\log \Delta})$.
\end{proof}

\noindent
Our next result says that the constrained ranking feasibility is hard even when we fix the structure of the set of properties; i.e., the only input to the problem are the upper-bound constraints. 
\begin{theorem}\label{thm:fixedhardness}
There exists a fixed family of properties $\mathcal{P}_m \subseteq 2^{[m]}$ for $m\in \N$ and  $p_m:=|\mathcal{P}_m| =O(m^2)$ such that the problem: ``given $n,m\in \N$ and a matrix of upper-bound values $U\in \N^{n \times p_m}$, check whether the constrained ranking instance defined by $\mathcal{P}_m$ and $U$ is feasible'', is $\NP$-hard.
\end{theorem}
\begin{proof}
Let the family of properties $\mathcal{P}_m$ be all 2-element subsets of $[m]$, i.e., 
\[\mathcal{P}_m =\{ \{i_1, i_2\}: i_1, i_2\in [m], i_1 \neq i_2\}.\]
We reduce the independent set problem to the above problem of constrained ranking maximization with a fixed property set. 

Consider any instance of the independent set problem, i.e., a graph $G=(V,E)$ and $t\in \N$ (the underlying question is: is there an independent set of size $n$ in $G$?).  
Construct a corresponding instance of constrained ranking maximization as follows: 
The number of ranking positions is $n$, and there are $m = |V|$ items to rank, one for every vertex $v\in V$.
Thus, there is a property for each pair of vertices. 
For every edge $e=\{i_1, i_2\} \in E$, set a fairness constraint that restricts the number of items having property $e$ in the top-$n$ positions to be at most $U_{n,e}=1$. 
In other words, at most one item out of any $\{i_1, i_2\} \in E$ can appear the ranking. 
For all non-edges, we leave the upper-bound constraint on the corresponding property unspecified, or equivalently, we let $U_{k,e}=k$ for all $e=\{i_1, i_2\}\notin E$. 

Clearly feasible rankings are in one-to-one correspondence with independent sets of size $n$ in $G$. 
Indeed, if we place items (vertices) from an independent set of size $n$ in the ranking in any order, then the resulting ranking is feasible, as the only essential constraints which are imposed on the ranking are that for every edge $e=\{i_1, i_2\} \in E$ at most one of $i_1, i_2$ is present in the ranking. 
Conversely, if these constraints are satisfied, the set of $n$ items placed in the ranking forms an independent set of size $n$. 
As the independent set problem is $\NP$-hard, via the above reduction we obtain hardness of the constrained ranking maximization problem with fixed properties.
\end{proof}

%%%%%%%%%%%%%%%%%%%%%%%

\noindent Lastly, one could hope that it is still possible solve the constrained ranking maximization problem by allowing constraints to be violated by a small amount.
However, this also remains hard.
The result states that finding a solution which violates all the constraints at most a factor of $c$ (for any constant $c$) is $\NP$-hard. 
Interestingly our proof relies on a certain bound from Ramsey theory. 
\begin{theorem}\label{thm:violation}
For every constant $c>0$, the following  violation gap variant of the constrained ranking feasibility problem (U) is $\NP$-hard.
\begin{enumerate}
%[topsep=1pt,itemsep=0ex,partopsep=0ex,parsep=1pt]
\item Output YES if the input instance is satisfiable.
\item Output NO if there is no solution which violates every upper-bound constraint at most $c$ times.  
\end{enumerate}
\end{theorem}
\begin{proof}
The reduction is inspired by an idea developed for the approximation hardness of packing integer programs by~\cite{CK99}. 
We use the inapproximabality of independent set (\cite{Hastad96,Zuckerman07}) to prove the theorem.  
It states that approximating the cardinality of the maximum independent set in an undirected graph $G=(V,E)$ to within a factor of $|V|^{1-\eps}$ is $\NP$-hard, for every constant $\eps>0$. 

Fix a constant $c>0$, and without loss of generality, assume that $c \in \N$. 
Consider the following reduction:  
given an instance of the independent set problem, i.e., a graph $G=(V,E)$ and a number $n \in \N$ the goal is to check whether there is an independent set of size $n$ in $G$. 
Let the set of items be $V$ and the number of positions in the ranking to be $n$. 
For every clique $C \subseteq V$ of cardinality $(c+1)$ in $G$, add a new property $C$ and set an upper-bound on the number of elements in $C$ in the top-$n$ positions of the ranking: $U_{n,C}=1$.
Note that these are at most ${m \choose c+1} \leq m^{c+1}=\poly(m)$ constraints. 

We claim the following:
\begin{enumerate}
%[topsep=1pt,itemsep=0ex,partopsep=0ex,parsep=1pt]
\item If there is a ranking that violates all of the constraints at most a factor of $c$, then there is an independent set of cardinality $\Omega\inparen{n^{1/(c+1)}}$ in $G$.
\item If there is no feasible ranking then there is no independent set of cardinality $n$ in $G$.
\end{enumerate}
Note that proving the above suffices to obtain the desired result: 
If there was a procedure to solve the constrained ranking feasibility problem for $c$, then one could  approximate in polynomial time the cardinality of the maximum independent set in a graph with an approximation ratio of $|V|^{1-\nfrac{1}{(c+1)}}$, which is not possible unless $P=NP$~\cite{Hastad96,Zuckerman07}.
Hence it remains to establish the claim. 

To establish the second claim, note that every independent set $S \subseteq V$ of size $n$ gives a feasible solution to the constrained ranking instance by placing the items corresponding to $S$ in the ranking in any order. 
To establish the first claim, suppose that there is a solution that violates every constraint by at most a factor of $c$. 
This means that we have a subset $S$ containing $n$ vertices in $G$ that does not contain any $(c+1)$-clique. 
Using a standard upper-bound on the Ramsey number $R \inparen{n^{\frac{1}{c+1}}, c}$ it follows that in the subgraph induced by $n$ in $G$ there exists an independent set of cardinality $\Omega\inparen{n^{1/(c+1)}}$.
\end{proof}

\section{Discussion and Future Work}\label{sec:discussion}

%\vspace{-0.15in}
In this paper, motivated by controlling and alleviating algorithmic bias in information retrieval, we initiate the study of  the complexity of a natural constrained optimization problem concerning rankings.
Our results indicate that the constrained ranking maximization problem, which is a generalization of the classic bipartite matching problem, shows fine-grained complexity. 
Both the structure of the constraints and the numbers appearing in upper-bounds play a role in determining its complexity.  Moreover, this problem generalizes several hypergraph matching/packing problems.
Our algorithmic results bypass the obstacles implicit in the past theory work by leveraging on the structural properties of the constraints and common objective functions from information retrieval.
More generally, our results not only contribute to the growing set of algorithms to counter algorithmic bias for fundamental problems \cite{Dwork2012,BS2015,Gummadi2015,FatML,CKSDKV17,celis2017complexity,kirkpatrick2016battling,CV17}, the structural insights obtained may find use in other algorithmic  settings related to the rather broad scope of ranking problems.

Our work also suggests some open problems and directions.
The first question concerns the $\Delta=1$ case and its (LU) variant; Theorem~\ref{thm:flow} implies that it can be solved in $\widetilde{O}(n^2m)$ time, can this be improved to nearly-linear time, as we do for the (U) variant (Theorem~\ref{thm:single})?
Another question is the complexity of the constrained ranking maximization problem (in all different variants) when $\Delta=2$ -- is it in $\mathbf{P}$?
The various constants appearing in our approximation algorithms  are unlikely to be optimal and improving them remains important.
{
In particular, our approximation algorithm for the case of large $\Delta$ in Theorem~\ref{thm:mult} may incur a $2$-multiplicative violation of constraints.
This could be significant when dealing with instances where the upper bound constraints are rather large (i.e., $U_{k\ell}\gtrsim \frac{k}{2}$) in which case, such a violation effectively erases all the constraints.
It is an interesting open problem to understand whether this $2$-violation can be avoided, either by providing a different algorithm or by making different assumptions on the instance.
}

{In this work we consider linear objective functions for the the ranking optimization problem, i.e., the objective is an independent sum of profits for individual item placements.
While this model might be appropriate in certain settings, there may be cases where one would prefer to measure the quality of a ranking as a whole, and in particular the utility of placing a given item at the $k$th position should also depend on what items were placed above it~\cite{ZCL2003}.
Thus defining and studying a suitable variant of the problem for a class of objectives on rankings that satisfy some version of the diminishing returns principle (submodularity), is of practical interest.
}

A related question that deserves independent exploration is to study the complexity of sampling a constrained ranking from the probability distribution induced by the objective (rather than outputting the ranking that maximizes its value, output a ranking with probability proportional to its value).
Finally, extending our results to the online setting seems like an important technical challenge which is also likely to have practical consequences.

\vspace{1mm}
\paragraph{Acknowledgments.} We would like to thank Mohit Singh and Rico Zenklusen for valuable discussions on the problem.

\bibliographystyle{alpha}
\bibliography{ranking}

\appendix

\section{Values $W_{ij}$ and Common Ranking Metrics}\label{sec:metrics}

In information retrieval, when proposing a new ranking method, it is naturally important to evaluate the quality of the resulting output. 
Towards this, a variety of ranking metrics have been developed. 
Ideally, one would use such metrics to compare how ``good'' a constrained ranking is (with respect to quality) in comparison to an unconstrained ranking. 
Putting it another way, we would like to maximize the quality of the constrained ranking as measured by such metrics. 
In this section we briefly survey three important classes of ranking metrics, translated to our setting, and show that our problem formulation can capture metrics by defining the values $W_{ij}$ appropriately.

Such metrics are often being defined with respect to the \emph{item quality} (often referred to as \emph{relevance} in the IR literature). 
Without loss of generality, we relabel the items so that $a_1 \geq a_2 \geq \cdots \geq a_m$ denote the qualities for $i \in \{1, \ldots, m\}$.
{For ease of presentation, we introduce simple forms of these metrics below; often in practice they are normalized for better comparison across rankings -- the two are equivalent with respect to optimization.}
We consider integral rankings $x \in \overline\Omega_{m,n}$ where $x\in \{0,1\}^{m \times n}$ if and only if for every $i\in [m]$ and every $j\in [n]$
\begin{equation}
\sum_{j=1}^n x_{ij} \leq 1 ~~~~\mbox{and}~~~~ \sum_{i=1}^m x_{ij}=1.
\end{equation}

\subsection{Rank-1 Metrics:} The quality of the ranking is considered to be the sum of the quality of its items where the quality of the item at position $j$ is \emph{discounted} by a value $f(j)$ that is non-increasing in $j$; in other words, getting the order correct at the top of the list is more important than getting it correct at the bottom. 
Formally, given a ranking $x \in \overline\Omega_{m,n}$, a rank-1 metric is defined as  
 \[m_{\mathrm{R1}}(x) := \sum_{j=1}^n \sum_{i : x_{ij} = 1} a_{i} \cdot f(j)\] 
 for a non-increasing function $f: [n] \to \R_{\geq0}$. 
Common examples are Discounted Cumulative Gain (DCG) \cite{jarvelin2002cumulated} where $f(i) = \frac 1 {\log(i+1)}$ for all $i$, and its variants.
For rank-1 metrics, we can simply define $W_{ij} := a_i \cdot f(j)$.
Let us now show that property \eqref{eq:monge} holds for such a metric.
Select any two $1\leq i_1 < i_2 \leq m $ and $1\leq j_1 < j_2 \leq n$. 
By definition $a_{i_1} \geq a_{i_2}$ and $f(j_1) \geq f(j_2)$, and hence $W_{i_1j_1} \geq W_{i_2j_1}$ and $W_{i_1j_1} \geq W_{i_1j_2}$.
Further, this implies that
\begin{align*}&(W_{i_1j_1} + W_{i_2j_2}) - (W_{i_1j_2} + W_{i_2j_1})\\
= &(a_{i_1}f(j_1) + a_{i_2}f(j_2)) -  (a_{i_1}f(j_2) + a_{i_2}f(j_1))\\
= &(a_{i_1} - a_{i_2})(f(j_1) - f(j_2))
\geq 0,
\end{align*}
so the property holds.

\subsection{Bradley-Terry Metrics:} 
There are a variety of multiplicative ranking metrics, originally developed with the goal of sampling from a distribution of (complete) rankings $x \in \overline\Omega_{m,m}$. 
For example, in the Bradley-Terry model \cite{bradley1952rank}, the quantity $\frac {a_{i_1}}{a_{i_1} + a_{i_2}}$ denotes the probability that item $i_1$ should be ranked above item $i_2$. 
This gives rise to the metric $m_{\mathrm{BT}}(x) := \prod\limits_{i_1,i_2: x_{i_1j_1} = x_{i_2j_2} = 1, j_1 < j_2} \frac {a_{i_1}}{a_{i_1} + a_{i_2}}$.
Removing the denominator (which is the same for all $x$ as every pair of items appears exactly once), it is easy to see that the metric can be rewritten as 
\[m_{\mathrm{BT}}(x) = \prod_{i : x_{ij} = 1} a_i^{m-j}.\]
Hence, it is equivalent to maximize $\log(m_{\mathrm{BT}}(x)) = \sum_{j=1}^n \sum_{i : x_{ij} = 1} (m-j) \log(a_i)$. 
By letting $a_i^\prime := \log(a_i)$ and $f(j) := (m-j)$, this is a rank-1 metric with non-increasing $f$ and $a_i^\prime$, and hence fits our formulation as discussed above. 

\subsection{Alignment Metrics:} 
One can also consider metrics that depend only on the difference in the \emph{item positions} between two rankings. 
Let $x^\star \in \overline\Omega_{m,m}$ be the optimal unconstrained ranking (this corresponds to sorting $i$s in non-increasing order of $a_i$). 
Alignment metrics are defined with respect to the positions of the items in $x^\star$, and we give some examples below. 

Given a ranking $x \in \overline\Omega_{m,n}$, in our setting 
Spearman's footrule  \cite{spearman1906footrule} corresponds to 
\[m_f(x) := \sum_{j=1}^n \sum\limits_{i: x_{ij} = x^\star_{ij^\star}=1}\left((2m-i-j) - |j - j^\star| \right),  \]
and similarly Spearman's rho \cite{spearman1904proof} corresponds to 
\[m_\rho(x) := \sum_{j=1}^n \sum\limits_{i: x_{ij} = x^\star_{ij^\star}=1} \left((2m-i-j)^2 - (j - j^\star)^2 \right).\]
These vary slightly from the usual definitions as it is for a \emph{partial} rather than complete ranking, and we measure \emph{similarity} rather than difference. The former is in line with partial ranking versions of Spearman's footrule and Spearman's rho as in \cite{fagin2003comparing,fagin2006comparing}, and for the latter we add a value to maintain non-negativity and monotonicity.

The above metrics can be naturally captured in our formulation by letting $W_{ij} := (2m-i-j) - |j - j^\star|$ or $W_{ij} := (2m-i-j)^2 - (j - j^\star)^2$ respectively where $j^\star$ is such that $x^\star_{ij^\star} = 1$.
Here, for any two $1\leq i_1 < i_2 \leq m $ and $1\leq j_1 < j_2 \leq n$ it is clear that $W_{i_1j_1} \geq W_{i_2j_1}$ and $W_{i_1j_1} \geq W_{i_1j_2}$ as desired.
Furthermore, it is not difficult to show via a simple but tedious case analysis of the possible orderings of $j_1, j_2, j^\star_1, j^\star_2$ that the desired property $(W_{i_1j_1} + W_{i_2j_2}) - (W_{i_1j_2} + W_{i_2j_1}) \geq 0$ holds.
One can also think of weighted versions of these metrics, as introduced in \cite{kumar2010generalized}, for which these observations generalize.

\end{document}